\documentclass[12pt, reqno]{amsart}
\usepackage{amssymb}
\usepackage{amsmath}
\allowdisplaybreaks[1]
\usepackage{graphicx,color}
\usepackage{wrapfig,framed,bbm}
\usepackage{mathtools}
\usepackage[height=22cm, width=16.5cm, hmarginratio={1:1}]{geometry}
\usepackage{hyperref}
\newcommand{\pa}{\partial}
\newcommand{\p}{\partial}
\newcommand{\nn}{\nonumber}
\newcommand{\e}{\epsilon}

\newcommand{\f}{\frac}

\newcommand{\la}{\lambda}
\newcommand{\bt}{{\bf t}}

\newcommand{\F}{\mathcal{F}}
\newcommand{\CC}{\mathbb{C}}
\newcommand{\ZZ}{\mathbb{Z}}
\newcommand{\beq}{\begin{equation}}
\newcommand{\eeq}{\end{equation}}
\newcommand{\half}{\frac{1}{2}}

\numberwithin{equation}{section}
\theoremstyle{plain}
\newtheorem{theorem}{Theorem}
\newtheorem{proposition}{Proposition}
\newtheorem{lemma}{Lemma}

\theoremstyle{definition}

\newtheorem{remark}{Remark}

\def\1{\mathbbm{1}}
\def\lan{\langle}
\def\ran{\rangle}

\begin{document}
\title[Virasoro-like algebra]{The Virasoro-like Algebra of a Frobenius Manifold}
\author[Liu, Yang, Zhang, Zhou]{Si-Qi Liu, Di Yang, Youjin Zhang, Jian Zhou}
\date{}
\maketitle

\begin{abstract}
For an arbitrary calibrated Frobenius manifold,
we construct an infinite dimensional Lie algebra, called the {\it Virasoro-like algebra}, which
is a deformation of the Virasoro algebra of the Frobenius manifold.
By using the Virasoro-like algebra we give a family of quadratic PDEs
that are satisfied by the genus-zero free energy of the Frobenius manifold.
We also derive, under the semisimplicity assumption,
the Virasoro constraints for the corresponding abstract Hodge partition function.
\end{abstract}

\setcounter{tocdepth}{1}
\tableofcontents

\section{Introduction}
Let 
$M$ be an $l$-dimensional Frobenius manifold
of charge~$d$ \cite{Du1,Du2}. 
Denote by $\eta$ the Gram matrix in a certain flat coordinate system of the invariant flat metric 
of the Frobenius manifold, and by $(\mu,R)$ the spectrum \cite{Du1,Du2,Du3,DZ-norm}. 
For the details about the constant matrices $\eta,\mu,R$ 
see Section~\ref{section2}.

In~\cite{DZ-norm} B.~Dubrovin and the third-named author of the present paper introduced 
the {\it regularized stress tensor} $T(\lambda;\nu)$ to construct the {\it Virasoro algebra}  
of~$M$ \cite{DZ2, EHX, Ge, LT}. Let us 
recall this construction. 
Denote by $a_{\alpha,p}$, $p\in \mathbb{Z}+1/2$, the linear operators given by
\beq\label{creation_annihilation}
a_{\alpha,p} = \left\{\begin{array} {cl}
\epsilon \f{\pa}{\pa
t^{\alpha,p-1/2}}, & p>0, \\
\\
\epsilon^{-1} (-1)^{p+1/2} \, \eta_{\alpha\beta} \,
\widetilde t^{\beta,-p-1/2}, &  p<0, \end{array}\right.
\eeq
where $t^{\alpha,k}$, $k\geq0$, are indeterminates (time variables of the Frobenius manifold), 
and $\widetilde t^{\alpha,k}=t^{\alpha,k}-\delta^{\alpha,1}\delta^{k,1}$.
Here and in what follows, free Greek indices
take the integer values from~1 to~$l$, and the
Einstein summation convention
is applied to repeated Greek indices with one-up and one-down.
The operators $a_{\alpha,p}$ for $p>0$ are called {\it annihilation operators} and for $p<0$ are called the {\it creation operators}. They satisfy the commutation relations
\[
[a_{\alpha,p},a_{\beta,q}] = (-1)^{p-1/2} \eta_{\alpha\beta} \delta_{p+q,0},\quad \forall \, p,q\in \ZZ+\frac12.
\]
Denote
\begin{align}
& f_\alpha(\la;\nu) := \int_0^\infty \f{dz}{z^{1-\nu}} e^{-\la z}
\sum_{p\in\mathbb{Z}+\f12} a_{\beta,p} \left(z^{p+\mu} z^R\right)^\beta_\alpha, \label{natural_with_nu}\\
&G^{\alpha\beta}(\nu) = -\f{1}{2\pi} \left[\left(e^{\pi i R}e^{\pi i (\mu+\nu)}+e^{-\pi i R}e^{-\pi i (\mu+\nu)}\right)\eta^{-1}\right]^{\alpha\beta}.\label{Galphabetanu0411}
\end{align}
The regularized stress tensor $T(\la;\nu)$ is then defined in~\cite{DZ-norm} as follows:
\begin{equation}
T(\la;\nu) = -\f12 : \pa_\la(f_\alpha(\la;\nu)) G^{\alpha\beta}(\nu) \pa_\la(f_\beta(\la;-\nu)):
+ \f{1}{4\la^2} {\rm tr} \left(\f{1}{4}-\mu^2\right). \label{Virasoro_field_with_nu}
\end{equation}
Here, ``$:~:$" denotes the {\it normal
ordering} (putting the annihilation operators always on the right of the creation operators).
The regularized stress tensor can be represented in the form~\cite{DZ-norm}
\begin{align}\label{TnuLmnu427}
& T(\la;\nu) = \sum_{m\in \ZZ} \frac{L_m(\nu)}{\lambda^{m+2}}.
\end{align}
Then, as it is shown in~\cite{DZ-norm}, the operators $L_m$, $m\ge-1$, defined by 
\beq\label{viralm0922}
L_m:=\lim_{\nu\to 0} L_m(\nu),
\eeq yield the {\it Virasoro operators} \cite{DZ2, EHX, Ge, LT}
of the Frobenius manifold~$M$. These operators 
satisfy the following Virasoro commutation relations:
\beq\label{Lmn0407}
\left[L_m,L_n\right] = (m-n) L_{m+n}, \quad \forall\,m,n\geq -1.
\eeq
The Virasoro algebra of~$M$ is an infinite-dimensional Lie algebra defined as 
$\bigoplus_{m=-1}^\infty \CC L_{m}$ with 
the commutator as the Lie bracket, denoted by~${\rm Vira}$. It  
 is used in~\cite{DZ-norm} to the reconstruction of higher genus free energies~\cite{DZ-norm} 
of~$M$ for the case when 
$M$ is semisimple. 

Let us proceed with the construction of the {\it Virasoro-like algebra}. Note that 
the $\nu$-dependence in~$L_m(\nu)$ for every $m\geq-1$~is {\it a priori} a series, but, we will
 show in Section~\ref{section2} that this series is actually {\it a polynomial}. 
 More interestingly, each polynomial~$L_m(\nu)$, $m\ge-1$, is
 {\it even} in~$\nu$ with degree $2[(m+1)/2]$, i.e.,
\beq\label{evenpoly}
L_m(\nu) =: \sum_{k=0}^{[(m+1)/2]} L_{m,2k} \nu^{2k}, \quad m\geq -1.
\eeq
The proof will again be given in Section~\ref{section2}. We call the operators with $L_{m,2k}$, $m\ge-1$,
$0\le k\le [(m+1)/2]$, the {\it Virasoro-like operators} of~$M$. In particular, $L_{m,0}=L_m=L_m(0)$, $m\geq-1$, 
are the Virasoro operators~\eqref{viralm0922}. 
In general, we have the following lemma. 
\begin{lemma} \label{linearindep923}
The operators ${\rm id}$, $L_{m,2k}$, $m\geq -1, 0\leq k\leq [(m+1)/2]$, are linearly independent. 
\end{lemma}
The proof of this lemma is given in Section~\ref{section3}. Let us then denote
\begin{align}
& {\rm Vira}_{\rm like}
:=\CC {\rm id} \bigoplus {\rm Span}_{\CC}  \left\{L_{m,2k} \mid m\geq -1, 0\leq k\leq [(m+1)/2]\right\}. \label{viralikedef0902}
\end{align}

\begin{theorem}\label{viralikealgebra}
${\rm Vira}_{\rm like}$ is a Lie algebra with the Lie bracket given by the commutator.
More precisely, there exist constants $c_{m,2k,n,2\ell,2h}$ $(m,n\geq -1$, $0\leq k \leq[(m+1)/2]$, $0\leq \ell \leq[(n+1)/2]$, 
$0\leq h\leq [(m+n+1)/2])$,
such that
\begin{align}
&\left[L_{m,2k},L_{n,2\ell}\right] = \sum_{h = 0}^{[(m+n+1)/2]}  c_{m,2k,n,2\ell,2h}  L_{m+n,2h} \nn\\
&-\delta_{m,1}\delta_{n,-1} \delta_{k,1} \delta_{\ell,0} \frac{l}2 + \delta_{m,-1}\delta_{n,1} \delta_{k,0} \delta_{\ell,1} \frac{l}2 . 
\label{structureviralike422-424}
\end{align}
Here, $m,n\geq -1$, $0\leq k\leq [(m+1)/2]$, $0\leq \ell\leq [(n+1)/2]$. 
Moreover, these constants $c_{m,2k,n,2\ell,2h}$ vanish whenever $h<k+\ell$. 
Furthermore, the constants
$c_{m,2k,n,2\ell,2h}$ are independent
of the Frobenius manifold. 
\end{theorem}

The proof of the above theorem will be given in Section~\ref{section3}.
We call ${\rm Vira}_{\rm like}$ the {\it Virasoro-like algebra} of~$M$,
and call $c_{m,2k,n,2\ell,2h}$ ($m,n\geq -1$, $0\leq k \leq[(m+1)/2]$, $0\leq \ell \leq[(n+1)/2]$, $k+\ell\leq h\leq [(m+n+1)/2]$)
the {\it essential structure constants} of the Virasoro-like algebra~${\rm Vira}_{\rm like}$.
An elementary description for these structure constants is given by Proposition~\ref{corollary10425} of Section~\ref{section3}.

The operator $L_{1,2}$ appeared in~\cite{EHX} and~\cite{LT} 
in the study of constraints for the genus zero free 
energy (see Sections~\ref{section2} and~\ref{section4}), and the $L_{2,2}$-operator appeared in~\cite{LT}. 
It was observed by B.~Dubrovin that the Virasoro-like algebra 
${\rm Vira}_{\rm like}$ could be 
generated by the operators $L_{1,2}$ and $L_{m,0}$, $m=-1,0,1,2$. 

It turns out that the Virasoro-like algebra contains at least {\it three} interesting Lie subalgebras. 
Let us describe these three. 
The first one is ${\rm Vira}$, i.e. the Virasoro algebra of~$M$, which is 
spanned by $L_{m,0}$, $m\geq-1$, as already given above in~\eqref{Lmn0407}.
The second Lie subalgebra of ${\rm Vira}_{\rm like}$, is spanned by 
${\rm id}$ and $L_m(1/2)$, $m\geq-1$, denoted by ${\rm Vira}_{1/2}$.
\begin{proposition}\label{virahalftwistthm}
The operators~$L_m(1/2)$ satisfy the following commutation relations:
\beq\label{commlmn12}
\left[L_m(1/2),L_n(1/2)\right] = (m-n) L_{m+n}(1/2) - \delta_{m,1} \delta_{n,-1} \frac{l}{8} +
 \delta_{m,-1} \delta_{n,1} \frac{l}{8}, \quad m,n\geq -1.
\eeq
\end{proposition}
The reason that we can take $\nu=1/2$ in~$L_m(\nu)$, $m\geq-1$, is because
$L_m(\nu)$ is a polynomial of~$\nu$.
The proof of Proposition~\ref{virahalftwistthm} is given in Section~\ref{section3}.
 \begin{remark}
We note that the Virasoro-like algebra ${\rm Vira}_{\rm like}$ is an infinite dimensional
deformation of ${\rm Vira}$ as well as of ${\rm Vira}_{1/2}$.
It is a reminiscence of the spectral flow of superconformal algebras \cite{CK,SS}.
In view of~\eqref{Lmn0407} and~\eqref{commlmn12}, it will be interesting
to investigate if the above two Lie subalgebras of ${\rm Vira}_{\rm like}$
could be related to the constructions in \cite{CK,SS} under a certain Bose--Fermi correspondence (cf.~also~\cite{Kac}).
\end{remark}
The third interesting Lie subalgebra of the Virasoro-like algebra ${\rm Vira}_{\rm like}$, is a {\it commutative} one, 
spanned by the operators $L_{2k-1,2k}$, $k\geq1$. 
These operators appeared in the work of Faber and Pandharipande~\cite{FP} in the study of Hodge integrals, and
 have the explicit expressions
\begin{equation}\label{L2k-1,2k}
L_{2k-1,2k} = \sum_{m\geq 0} \tilde{t}^{\alpha,m} \frac{\p}{\p t^{\alpha,m+2k-1}}
- \frac{\epsilon^2}{2} \sum_{m=0}^{2k-2} (-1)^m \eta^{\alpha\gamma} \frac{\p^2}{\p t^{\alpha,m}\p t^{\gamma,2k-2-m}}.
\end{equation}

Let us proceed to present several applications of the Virasoro-like algebra.
Fix a calibration $\theta_{\alpha,p}$, $p\geq0$, of the Frobenius manifold~$M$, 
and denote by~$\F_0(\bt)$ the genus zero free energy~\cite{Du1,DZ-norm} 
of the calibrated Frobenius manifold (see Section~\ref{section2}). 
Here $\bt:=(t^{\alpha,p})_{1\leq \alpha\leq l, p\geq 0}$.
The first application of the Virasoro-like algebra is given by the following theorem. 
\begin{theorem} \label{theoremcor415}
The following relations hold true as $\e\rightarrow 0$:
\beq\label{viralikeconstraints}
e^{-\e^{-2}\F_0(\bt)} L_{m,2k} e^{\e^{-2}\F_0(\bt)} = {\rm O}(1), \quad m\geq -1,\,0\leq k \leq [(m+1)/2] .
\eeq
\end{theorem}
The proof, based on the results of~\cite{DZ2}, will be given in Section~\ref{section4}.

For the second application we will assume that $M$ is {\it semisimple}.
Denote by~$Z$ the partition function of this Frobenius manifold~\cite{DZ-norm} (cf.~also~\cite{Du4,Givental}). Namely,
\beq
Z = Z(\bt; \e) = \sum_{g\geq0} \e^{2g-2} \mathcal{F}_g(\bt)
\eeq
is the power series of ${\bf t}_{>0}$ that satisfies the conditions
\begin{align}
& L_{m,0} Z = 0, \quad \forall\, m\geq -1,\nn\\
& \sum_{p\ge0} \tilde{t}^{\alpha,p} \frac{\p Z}{\p t^{\alpha,p}} + \e \frac{\p Z}{\p \e} + \frac{l}{24} Z = 0 , \nn\\
& \mathcal{F}_g (\bt) = F_g \left(v_{\rm top}(\bt), \p_x(v_{\rm top}(\bt)),\dots, \p_x^{3g-2} (v_{\rm top}(\bt))\right), \quad g\geq 1,\nn
\end{align}
where $v_{\rm top}(\bt)$ denotes the {\it topological solution} to the Principal Hierarchy of~$M$ (cf.~Section~\ref{section2}), $F_g=F_g(z_0,z_1,\dots,z_{3g-2})$, $g\geq1$, are polynomials of $z_2,z_3,\dots,z_{3g-2}$ 
and rational functions of~$z_1$ with
coefficients that are smooth functions of~$z_0$, and $z_k$ are vectors whose components are indeterminates. 
Here $\bt_{>0}:=(t^{\alpha,p})_{1\leq \alpha\leq l, p\geq 1}$.
The {\it abstract Hodge partition function}~$Z_{\rm H}$ associated with the Frobenius manifold~$M$ is defined 
as follows~\cite{FP,Givental} (cf.~also~\cite{DLYZ}):
\beq\label{hodgepartitionfunction420}
Z_{\rm H} = Z_{\rm H}(\bt; \boldsymbol{\sigma}; \e) := \exp\left(-\sum_{k\ge1} \frac{\sigma_{2k-1} B_{2k}}{(2k)!} L_{2k-1,2k} \right) \, Z.
\eeq
Here, $B_n$, $n\geq 0$, denotes the $n$th Bernoulli number.
In the case when $M$ comes from the quantum cohomology of some smooth projective variety~$X$, the partition function $Z$ coincides with the
partition function for the Gromov-Witten invariants of~$X$ \cite{DZ-norm, Givental, Teleman}, and the abstract Hodge
partition function $Z_{\rm H}$ 
coincides with the partition function of the Hodge integrals of~$X$ as it is shown~\cite{FP}.

\begin{theorem}\label{theorem30407}
The Hodge partition function satisfies the following Virasoro constraints:
\beq
L_n^{\rm H} Z_{\rm H}  = 0, \quad n\ge -1,
\eeq
with the operators $L_n^{\rm H}$ given in terms of the Virasoro-like operators by
\begin{align}
& L_{-1} ^{\rm H} =  L_{-1,0} + \frac{\sigma_1}{24} l , \label{lminus1421}\\
& L_n^{\rm H} = L_{n,0} + \sum_{m=1}^\infty \frac{(-1)^m}{m!} \sum_{k_1,\dots,k_m\geq 1}
\prod_{j=1}^m \frac{B_{2k_j} \sigma_{2k_j-1}}{\binom{n+1-j+2\sum_{i=1}^j k_i}{n+2-j+2\sum_{i=1}^{j-1} k_i}} \nn\\
& \times \sum_{h_1=k_1}^\infty \sum_{h_2=k_2+h_1}^\infty \dots \sum_{h_m=k_m+h_{m-1}}^\infty L_{n+(2k_1-1)+\cdots+(2k_m-1),2h_m} 
\prod_{j=1}^m \left\{\begin{array}{c} 2h_j-2h_{j-1} \\ 2k_j-1\end{array}\right\} \binom{2h_j}{2h_{j-1}},
\end{align}
where $\left\{\begin{array}{c} a \\ b\end{array}\right\}$ denotes the Stirling number of
the second kind $($cf.~\eqref{stirlinggen425}$)$, and 
$h_0$ is understood to be zero.
Moreover, the operators $L_n^{\rm H}$ satisfy the commutation relations
\beq
\left[L_m^{\rm H},L_n^{\rm H}\right] = (m-n) L_{m+n}^{\rm H}, \quad \forall\,m,n\ge-1.
\eeq
\end{theorem}

Inspired by Fabler--Pandharipande's formula~\eqref{hodgepartitionfunction420}, 
still under the assumption that $M$ is semisimple, let us consider 
a further deformation of the abstract Hodge partition function and the corresponding integrable hierarchy. 
Recall that the {\it Hodge integrable hierarchy} (for short, Hodge hierarchy) of~$M$ is constructed in~\cite{DLYZ},
such that the partition function~$Z_{\rm H}$ is a particular tau-function of the Hodge hierarchy.
Based on the genus zero Virasoro-like constraints~\eqref{viralikeconstraints}, we can define the following power series
with infinitely many parameters:
\beq
\label{viralikepartitionfunction0902}
Z_{\rm like} = Z_{\rm like}(\bt; \boldsymbol{r}; \e) := \left(\prod_{m=1}^\infty \prod_{1\le k \le [(m+1)/2]} \exp \left(r_{m,2k} L_{m,2k}\right)\right)  Z.
\eeq
Here, $\boldsymbol{r}=(r_{1,2},r_{2,2},r_{3,2},\dots)$ denotes an infinite vector of indeterminates,
the products mean operator compositions, and we note that unlike~\eqref{hodgepartitionfunction420} the operators
$L_{m,2k}$ are in general non-commutative~\eqref{structureviralike422-424}, and so in the above products we fix an order of the compositions as
follows: $$(m,k)=(1,2), (2,2), (3,2), (3,4), (4,2), (4,4), (5,2), (5,4), (5,6), \, {\rm ect.}$$
(We note that the definition of $Z_{\rm like}$ depends on the order of the compositions; 
the above order-fixing is just for simplicity.)
The following theorem will be proved in a forthcoming joint work with B.~Dubrovin and P.~Rossi.

\smallskip

\noindent {\bf Theorem A}.
{\it The power series $Z_{\rm like}$ is a particular tau-function of the topological solution to a
 not necessarily Hamiltonian tau-symmetric integrable system, which is an integrable deformation
 of the Principal Hierarchy depending on an infinite family of parameters $r_{m,2k}$, $m\geq1$, $k=1,\dots, [(m+1)/2]$.} 

\medskip

\noindent {\bf Organization of the paper.} In Section~\ref{section2},
we give a review on Frobenius manifolds
and present some properties of the regularized stress tensor $T(\lambda;\nu)$.
In Section~\ref{section3} we compute the commutators between Virasoro-like operators, prove Theorem~\ref{viralikealgebra},
and give descriptions about the essential structure constants of the Virasoro-like algebra.
In Section~\ref{section4} we prove Theorem~\ref{theoremcor415}.
In Section~\ref{section5} we prove Theorem~\ref{theorem30407}.

\smallskip

\noindent {\bf Acknowledgements.}
We wish to thank Boris Dubrovin for valuable suggestions.
Part of the work of D.Y. was done during his PhD studies at Tsinghua University;
he thanks Tsinghua for excellent working conditions and financial support. 
The work is partially supported by the National Key R and D Program of China 2020YFA0713100, 
and by NSFC grants 11725104, 12171268, 11661131005, 11890662.

\section{Properties of the regularized stress tensor}\label{section2}
In this section, we review the definition and the spectrum data around zero of a Frobenius manifold, 
and then study the regularized stress tensor $T(\lambda;\nu)$ in more details.

\subsection{Review on Frobenius manifolds}
A {\it Frobenius algebra (FA)} is a triple 
$\bigl(A^l, e,  \eta\bigr)$, where $A$ is a commutative and associative
algebra over~$\CC$ with the unity~$e$,  and $\eta$ is a symmetric and non-degenerate
bilinear form $A\times A\rightarrow \CC$ satisfying the condition
\[\eta\left(x\cdot y, z\right) = \eta\left(x, y \cdot z\right), \quad \forall \, x,y,z \in A.\]
A {\it Frobenius structure of charge~$d$} \cite{Du1,Du2} on a complex manifold~$M^l$ is a family of
Frobenius algebras $\bigl(A_p= T_p M, e_p, \eta_p\bigr)$, $p\in M$,
with the multiplication~``$\,\cdot_p$", the unity $e_p$ and the bilinear form $\eta_p$
depending {\it holomorphically} on~$p$, satisfying the following axioms:
\begin{itemize}
\item[1.] The metric $\eta$ is {\it flat}. Denote by~$\nabla$ the Levi--Civita connection with respect to~$\eta$.
It is required that
$ \nabla e = 0$.
\item[2.] Define a 3-tensor field by $c(X,Y,Z):= \eta \left(X\cdot Y, Z\right)$ for $X,Y,Z$ being arbitrary
 holomorphic vector fields on~$M$.
The 4-tensor field~$\nabla c$ is {\it totally symmetric}.
\item[3.] There exists a holomorphic vector field $E$ satisfying
\begin{align}
&\nabla\nabla E = 0 , \label{aE1}\\
& \left[E,  X \cdot Y\right] - \left[E, X\right] \cdot Y - X \cdot \left[E, Y\right] = X \cdot Y , \label{aE2} \\
& E \lan X, Y\ran - \lan [E,X] , Y\ran - \lan X , [E,Y]\ran = (2-d) \lan X, Y\ran. \label{aE3}
\end{align}
\end{itemize}
A complex manifold equipped with a Frobenius structure is called a {\it Frobenius manifold}. 
We assume in this paper that $\nabla E$ is {\it diagonalizable}.
The
{\it extended deformed flat connection}~\cite{Du1,Du2}, denoted by~$\widetilde\nabla$, 
is an affine connection on~$M\times \CC^*$,
defined by
\begin{align}
& \widetilde \nabla_X  Y  =  \nabla_X Y + z X\cdot Y,\\
& \widetilde \nabla_{\p_z} Y = \frac{\p Y}{\p z} + E \cdot Y - \frac1z \left(\frac{2-d}2 1 - \nabla E \right)Y,\\
& \widetilde \nabla_{\p_z}  \p_z = \widetilde \nabla_X \p_z = 0,
\end{align}
where $X,Y$ are arbitrary horizontal holomorphic vector fields on $M\times\mathbb{C}^*$, $z\in\CC^*$.
The definition of a Frobenius manifold implies that  $\widetilde \nabla$ is flat~\cite{Du1}.

Take $v=(v^1,\dots,v^l)$
a system of local coordinates satisfying  
\begin{align}
& \eta\left(\frac{\p}{\p v^\alpha},\frac{\p}{\p v^\beta}\right) =: \eta_{\alpha\beta} ~ \mbox{ are constants}, \label{1005flat}\\
& e = \frac{\p}{\p v^1}, \quad
E = \sum_{\beta=1}^l \left(1-\frac d2-\mu_\beta\right) v^\beta \frac{\p}{\p v^\beta} + \sum_{\beta=1}^l r_\beta \frac{\p}{\p v^\beta}, \label{eEv}
\end{align}
where $\mu_\beta$, $r_\beta$ are constants. 
(Local coordinates satisfying~\eqref{1005flat} are called {\it flat coordinates} of~$\eta$.)
The existence for the choice of flat coordinates satisfying~\eqref{eEv} follows from the definition of a Frobenius manifold~\cite{Du1,Du2}.
Denote
$c_{\alpha\beta\gamma} := \eta (\p_{v^\alpha} \cdot \p_{v^\beta} , \p_{v^\gamma} )$,
and we will use $(\eta_{\alpha\beta})$ and
its inverse $(\eta^{\alpha\beta}):=(\eta_{\alpha\beta})^{-1}$ to lower and raise the Greek indices, respectively. For example,
$v_\alpha:=\eta_{\alpha\beta}v^\beta$, $c_{\alpha\beta}^\gamma:=\eta^{\gamma\rho} c_{\rho\alpha\beta}$, etc.

A holomorphic function $f=f(v;z)$ defined on a certain open subset of $M\times \CC^*$ is called {\it $\widetilde \nabla$-flat}, if 
$\widetilde \nabla df=0$. The flatness of~$\widetilde\nabla$
ensures the local existence of $n$ linearly independent $\widetilde \nabla$-flat holomorphic functions
$\tilde v_\alpha(v;z)$ on $M \times \mathbb{C}^*$, which will be called the {\it deformed flat coordinates}. 
Let us recall their construction. For a $\widetilde \nabla$-flat holomorphic function
$f(v;z)$ on $M \times \mathbb{C}^*$, denote $y^\alpha=\eta^{\alpha\beta} \p f/\p v^\beta$.
Then the vector-valued function $y:=(y^1,\dots,y^l)^T$ satisfies the following system of linear equations:
\begin{align}
\frac{\p y}{\p v^\alpha }  & =  z C_\alpha y, \label{linear1} \\
\frac{d y}{d z} & = \left(\mathcal{U} + \frac{\mu}{z}\right) y,\label{linear2}
\end{align}
where $C_\alpha$, $\mathcal{U}$ and $\mu$ are matrices defined by
$C_\alpha =  \left(c_{\alpha\beta}^\gamma\right)$,
$\mathcal{U} = \left(E^\gamma c_{\gamma\beta}^\alpha\right)$,
$\mu = {\rm diag}\left(\mu_1,\dots,\mu_l\right)$.

Observe that the ODE system~\eqref{linear2} has two singularities on the complex $z$-plane: $z=0$ and $z=\infty$.
The singularity $z=0$ is Fuchsian. It is shown~\cite{Du1,Du2} that
there exists a fundamental solution matrix~$\mathcal{Y}$
to equations \eqref{linear1}--\eqref{linear2} of the following form:
\begin{align}
\mathcal{Y} =  \Phi (v;z) z^\mu z^R, \quad v \in M, z \in \CC^*, \label{yab0411}
\end{align}
where 
\[\Phi (v;z)=\sum_{k\geq 0} \Phi_k(v) z^k 
\]
is an analytic matrix-valued function on $M\times \CC$ satisfying 
the normalization condition
\begin{equation} \label{yb}
\Phi_0(v) \equiv I, \quad 
\eta^{-1} \Phi(v;-z)^T \eta \Phi(v;z) \equiv I, 
\end{equation}
and~$R$ is a constant matrix satisfying the conditions
\begin{align}
& z^\mu R z^{-\mu} =: \sum_{s \geq 1} R_s z^s \quad (\mbox{in particular } R=R_1+R_2+\cdots),\\
& \eta^{-1} R_s^T \eta = (-1)^{s+1} R_s, \quad s\ge1,\\
& (R_s)^\alpha_\beta \neq 0  \quad \mbox{only if } ~ \mu_\alpha-\mu_\beta = s, \quad  s\ge1.
\end{align}
The matrix~$R$ may not be unique (cf.~\cite{Du1,Du2} or
Section~\ref{section2}). The matrices $\mu,R$ are called the {\it spectrum data} around $z=\infty$ of the Frobenius 
manifold. 
From \cite{Du1,Du2} we also know that locally there exist analytic functions $\theta_\alpha(v;z)$
on $M \times \mathbb{C}$ such that
\beq
\Phi^{\alpha}_\beta(v;z) = \eta^{\alpha\gamma} \frac{\p \theta_\beta(v;z)}{\p v^\gamma}.
\eeq
So the functions $\tilde v_\alpha(v;z)$ defined via
\[
\left(\tilde v_1(v;z), \dots, \tilde v_l(v;z)\right) = \left(\theta_1(v;z), \dots, \theta_l(v;z)\right) z^\mu z^R
\]
form a system of deformed flat coordinates for~$M$.
Write $\theta_\alpha(v;z)=:\sum_{k\ge0} \theta_{\alpha,k}(v) z^k$, then we have
\begin{align}
& \p_\alpha\p_\beta (\theta_{\gamma,k+1}) = c_{\alpha\beta}^\sigma \p_\sigma (\theta_{\gamma,k}), \label{theta1409}\\
& \frac{\p\theta_{\alpha,0}}{\p v^\beta} = \eta_{\alpha\gamma}, \label{theta2409}\\
& \eta \left(\nabla \theta_\alpha(v;z) , \nabla \theta_\beta(v;-z) \right) = \eta_{\alpha\beta} , \label{theta3409}\\
& \mathcal{L}_E \left(\frac{\p \theta_{\alpha,k}}{\p v^\beta}\right) =
\left(k+\mu_\alpha+\mu_\beta\right) \frac{\p \theta_{\alpha,k}}{\p v^\beta}
+ \sum_{r=1}^k (R_r)^\gamma_\alpha \frac{\p \theta_{\gamma,k-r}}{\p v^\beta}. \label{theta4409}
\end{align}
We can further normalize $\theta_{\alpha,k}$ as follows:
\beq\label{theta5409}
\theta_{\alpha,0} = v_\alpha, \qquad  \frac{\p \theta_{\alpha,k+1}(v)}{\p v^1} = \theta_{\alpha,k}(v), \quad k\geq0.
\eeq
A choice of~$\{\theta_{\alpha,k}\}_{k\geq0}$ satisfying the above~\eqref{theta1409}--\eqref{theta5409} is called a
{\it calibration} (cf.~\cite{DLYZ}) on~$M$.
We will fix a particular calibration, and the Frobenius manifold~$M$ is now
by default calibrated.

The \emph{Principal Hierarchy}~\cite{Du1,DZ-norm} of~$M$ is defined by
\beq\label{principal hierarchy}
\f{\pa v^\alpha}{\pa t^{\beta,k}} = \eta^{\alpha\gamma} \p_x\left(\f{\p\theta_{\beta,k+1}}{\pa v^\gamma}\right),\quad k\ge 0,
\eeq
where $t^{\beta,k}$ are time variables with~$t^{1,0}$ identified with~$x$.
We are interested in solutions $v=v(\bt)$ to the Principal Hierarchy 
that are power series of $\bt_{>0}$. 
The {\it topological solution} $v_{\rm top}(\bt)$ to the Principal Hierarchy is defined 
as the unique power-series-in-$\bt_{>0}$ solution satisfying 
\beq\label{topoini1005}
v_{\rm top}^\alpha(\bt) \big|_{\bt_{>0}={\bf 0}} = t^{\alpha,0}.
\eeq
Alternatively, $v_{\rm top}(\bt)$ can also be determined by the {\it genus zero Euler--Lagrange equation}
\beq
\sum_{k\geq 0} \widetilde{t}^{\alpha,k} \frac{\p \theta_{\alpha,k}}{\p v^\gamma} = 0.
\eeq
It is indicated in~\cite{DZ2} (cf.~the Section~3 of~\cite{DZ2}) that the coefficients of monomials of $\bt_{>0}$ 
of $v_{\rm top}^\alpha(\bt)$ 
are locally holomorphic functions; this allows one to do Taylor expansions of these coefficients 
at different particular points.

Following~\cite{Du1,Du4,DZ-norm},
define a family of analytic functions $\Omega_{\alpha,i;\beta,j}(v)$ on~$M$ via
\beq
\sum_{i,j\geq 0} \Omega_{\alpha,i;\beta,j}(v) z^i w^j =
\frac{\left\langle \nabla\theta_\alpha(v;z), \nabla\theta_{\beta}(v;-w) \right\rangle-\eta_{\alpha\beta}}{z+w}.
\eeq
The (topological) {\it genus zero free energy}~$\mathcal{F}_0=\mathcal{F}_0(\bt)$ associated with~$M$ is then defined as
a particular power series of ${\bf t}_{>0}$, explicitly given by
\beq\label{F_0}
\mathcal{F}_0(\bt) :=
\frac12 \sum_{i,j\geq 0} \widetilde t^{\alpha,i} \widetilde t^{\beta,j} \Omega_{\alpha,i;\beta,j}(v_{\rm top}(\bt)).
\eeq

In the next subsection we will derive explicit expressions of~$f_\alpha(\la;\nu)$ and~$T(\la;\nu)$
defined in the introduction. 

\subsection{Properties of $L_m(\nu)$}
Before doing the derivations, 
it is helpful to recall from \cite{DZ2,DZ-norm} some useful properties of $(\mu,R)$. First, any polynomial~$P$ of $R_1,R_2,\dots$
can be decomposed uniquely in the following way:
\begin{align}
& P(R_1,R_2,\dots) = \sum_{r\ge0} \left(P(R_1,R_2,\dots)\right)_{r}, \label{PRdecomp430}\\
& z^\mu (P(R_1,R_2,\dots))_r z^{-\mu} = z^r \left(P(R_1,R_2,\dots)\right)_{r},
\end{align}
where the summation on the right-hand side of~\eqref{PRdecomp430} only contains a finite number of terms.
Other useful properties for $(\mu,R)$ are listed here:
\begin{align}
& \eta^{-1} \mu^T \eta = - \mu, \label{p10411}\\
& \left[\mu, \left(P(R_1,R_2,\dots)\right)_r\right] = r \left(P(R_1,R_2,\dots)\right)_r, \label{p0924}\\
& f(\mu) \left(P(R_1,R_2,\dots)\right)_r = \left(P(R_1,R_2,\dots)\right)_r f(\mu+r), \quad \forall\,r\geq0, \label{p20411}\\
& \eta^{-1} \left(\left(R^k\right)_r\right)^T \eta = (-1)^{k+r} \left(R^k\right)_r  , \quad \forall\,k,r\geq 0,\label{p30411} \\
& \lambda^{-R^*} = e^{\pi i \mu} \lambda^{R} e^{-\pi i \mu} = e^{-\pi i \mu} \lambda^{R} e^{\pi i \mu} \label{p40415}
\end{align}
for an arbitrary power series~$f(z)$ and an arbitrary polynomial $P(R_1,R_2,\dots)$.
Denote 
\[
a_p := (a_{1,p},\dots,a_{l,p}), \quad a^p := \eta^{-1} a_p^T, \qquad p\in \ZZ+\f12,
\]
then we have the results of the following two lemmas that are given in~\cite{DZ-norm}; see the proof of Theorem 3.7.11 and Lemma 4.1.7 therein.

\begin{lemma}[\cite{DZ-norm}] The series $f_\alpha(\la,\nu)$ has the following explicit expression:
\begin{equation}\label{natural_with_nu_explicit}
f_\alpha(\la,\nu) = \sum_{p\in\mathbb{Z}+\f12 }
a_{\beta,p} \sum_{r\geq 0} \left(\left(e^{R\pa_{\nu}}\right)_r \left(\Gamma(\mu+\nu+p+r)\right) \la^{-(\mu+p+r+\nu)}\la^{-R}\right)^\beta_\alpha .
\end{equation}
\end{lemma}

\begin{lemma}[\cite{DZ-norm}] The series $L_m(\nu)$, $m\in\ZZ$, has the following expression:
\begin{equation}\label{Virasoro-like_op_with_nu}
L_m(\nu) = \f12  \sum_{p,q\in\mathbb{Z}+\f12, r\geq 0 \atop p+q+r=m} 
: a_p N^p_q(r,\nu)  a^q : + \f{\delta_{m,0}}{4} {\rm tr} \left(\f{1}{4}-\mu^2\right),
\end{equation}
where
\begin{equation}\label{N}
N^p_q(r,\nu) := \f1{\pi} \left( \left(e^{R\pa_\nu}\right)_r \left(
\Gamma(\mu+\nu+p+r+1)  \cos\pi(\mu+\nu)  \Gamma(-\mu-\nu+q+1)\right)\right).
\end{equation}
\end{lemma}

We also have the following lemmas on further properties of the regularized stress tensor.
\begin{lemma}\label{symmetry_N}
For $p,q\in \ZZ+\f12$, the matrix-valued function $N^p_q(r,\nu)$ satisfies
the identity
\begin{equation}
\eta^{-1} N^p_q(r,\nu)^T \eta = N^q_p(r,-\nu).
\end{equation}
\end{lemma}
\begin{proof} By using~\eqref{N} and the properties~\eqref{p10411}--\eqref{p30411} we have
\begin{align}
&\eta^{-1} N^p_q(r,\nu)^T \eta \nn\\
&\quad =\f{\eta^{-1}}{\pi}  \left\{\sum_{k=0}^\infty
\f{\left(R^k\right)_r}{k!} \pa_\nu^k \left(\Gamma(\mu+\nu+p+r+1)\cos\pi(\mu+\nu)\Gamma(- \mu-\nu+q+1)\right)\right\}^T \eta \nn \\
&\quad =\f1{\pi}  \sum_{k=0}^\infty \f{\pa_\nu^k}{k!}
\left(\Gamma( \mu-\nu+q+1)\cos\pi(-\mu+\nu)\Gamma(- \mu+\nu+p+r+1)\right)(-1)^{k+r} \left(R^k\right)_r \nn \\
&\quad =\f1{\pi} \sum_{k=0}^\infty \f{\left(R^k\right)_r(-\pa_\nu)^k}{k!}\left(\Gamma( \mu-\nu+q+r+1)\cos\pi(-\mu+\nu-r)\Gamma(- \mu+\nu+p+1)\right)(-1)^r \nn \\
&\quad =\f1{\pi}  \left(e^{-R\pa_\nu}\right)_r
\left(\Gamma( \mu-\nu+q+r+1)\cos\pi(\mu-\nu)\Gamma(- \mu+\nu+p+1)\right),\nn
\end{align}
which proves the lemma.
\end{proof}

\begin{lemma}\label{symmetry}
For every $m\in\mathbb{Z}$, $L_m(\nu)$ is invariant under $\nu\mapsto -\nu$, i.e.,
\beq
L_m(\nu) = L_m(-\nu).
\eeq
\end{lemma}
\begin{proof} We have
\begin{align}
L_m(-\nu)
&=\f12\sum_{p\in\mathbb{Z}+\f12 \atop r\ge0 }:a_p  \eta^{-1} \eta  N^p_{m-r-p}(r,-\nu)  \eta^{-1}  a_{m-r-p}^T: \nn \\
&=\f12\sum_{p\in\mathbb{Z}+\f12 \atop r\ge0 }:a_p  \eta^{-1} (N_p^{m-r-p}(r,\nu))^T  a_{m-r-p}^T: \nn \\
&=\f12\sum_{p\in\mathbb{Z}+\f12 \atop r\ge0 }:a_{m-r-p} N^{m-r-p}_{p}(r,\nu)  \eta^{-1}  a^T_p:  \; = L_m(\nu). \nn
\end{align}
The lemmas is proved.
\end{proof}
From Lemma~\ref{symmetry} and equation~\eqref{N} we know that \eqref{evenpoly} holds true.

By using
Lemma~\ref{symmetry_N}
we can write $L_m(\nu)$, $m\geq-1$, into the following convenient form:
\begin{equation}\label{L_m_R}
L_m(\nu) = \f12 \sum_{p,q\in\mathbb{Z}+\f12, r\geq 0 \atop p+q+r=m}
: a_p M^p_q(r,\nu) a^q : + \f{\delta_{m,0}}{4} {\rm tr} \left(\f{1}{4}-\mu^2\right),
\end{equation}
where $M^p_q(r,\nu)$ is defined by
\begin{equation}\label{M}
M^p_q(r,\nu) = \f12 \left(N^p_q(r,\nu)+N^p_q(r,-\nu)\right).
\end{equation}

\begin{lemma}
For arbitrary $p,q\in\mathbb{Z}+\f12$ the following identity holds true:
\begin{equation}\label{etaMtransp413}
\eta^{-1} M^p_q(r,\nu)^T \eta = M^q_p(r,\nu).
\end{equation}
\end{lemma}
\begin{proof} We have
$$\eta^{-1} M^p_q(r,\nu)^T \eta =
\f12 \eta^{-1} \left(N^p_q(r,\nu)^T+N^p_q(r,-\nu)^T\right) \eta = \f12 \left( N^q_p(r,-\nu)+N^q_p(r,\nu) \right),$$
which proves the lemma.
\end{proof}

\section{The Virasoro-like algebra}\label{section3}

In this section we prove Theorem~\ref{viralikealgebra} and Proposition~\ref{virahalftwistthm}.

Let us first recall some notations. The {\it unsigned Stirling numbers of the first kind}
$\left[\begin{array}{c} n \\ k \\ \end{array} \right]$
are integers, defined via the generating function
\beq
(x)_n = \sum_k \left[\begin{array}{c} n \\ k \\ \end{array} \right] x^k, \quad n\geq 0, \label{stirling424}
\eeq
where $(x)_n:=x(x+1)\cdots(x+n-1)$ denotes the {\it increasing Pochhammer symbol}. The
{\it Stirling numbers of the second kind} $\left\{\begin{array}{c} n \\ k \\ \end{array} \right\}$ are integers,
that can be defined via the generating function
\beq\label{stirlinggen425}
x^n = \sum_k (-1)^{n-k} \left\{\begin{array}{c} n \\ k \\ \end{array}\right\} (x)_{k}, \quad n\geq0.
\eeq
Alternatively, $\left\{\begin{array}{c} n \\ k \\ \end{array} \right\}$ admit the following generating function:
\beq
\sum_{n} \left\{\begin{array}{c} n \\ k \\ \end{array}\right\} \frac{x^n}{n!} = \frac{(e^x-1)^k}{k!}.
\eeq
For more details about these numbers (for instance their combinatorial meanings) see~\cite{GKP}.

We now consider the commutation relation between the Virasoro-like operators.
\begin{lemma} For arbitrary $w\in\mathbb{Z}+\f12$, $m\geq-1$,
\begin{equation}\label{comm_a_L}
\left[a_{\alpha,w}, L_m(\nu)\right] = (-1)^{w-\f12} \sum_{r\geq 0} a_{\beta,m-r+w} \left(M^{m-r+w}_{-w}(r,\nu)\right)^\beta_\alpha .
\end{equation}
\end{lemma}
\begin{proof}
We have
\begin{eqnarray}
&& [a_{\alpha,w},L_m(\nu)]\notag\\
&=& \f12 \Biggl[a_{\alpha,w}, \sum_{p,q\in\mathbb{Z}+\f12,r\geq 0 \atop p+q+r=m}
: a_{\beta,p} \left(M^p_q(r,\nu)\right)^\beta_\gamma \eta^{\gamma\epsilon} a_{\epsilon,q} : \Biggr]\notag\\
&=& \f12 \Biggl[a_{\alpha,w}, \sum_{p,q\in\mathbb{Z}+\f12,r\geq 0 \atop p+q+r=m}
a_{\beta,p} \left(M^p_q(r,\nu)\right)^\beta_\gamma \eta^{\gamma\epsilon} a_{\epsilon,q} \Biggr]\notag\\
&=&\f{(-1)^{w-\f12 }}{2} \Biggl(\sum_{q\in\mathbb{Z}+\f12 ,r\geq 0 \atop q+r=m+w }
\eta_{\alpha\beta} \left(M^{-w}_q(r,\nu)\right)^\beta_\gamma \, \eta^{\gamma\epsilon} a_{\epsilon,q} +
\sum_{p\in\mathbb{Z}+\f12 ,r\geq 0\atop p+r=m+w} \left(M^p_{-w}(r,\nu)\right)^\beta_\alpha a_{\beta,p} \Biggr)\notag\\
&=&\f{(-1)^{w-\f12 }}{2}\left(\sum_{r\geq 0} a_{\beta,m-r+w} \left(M^{m-r+w}_{-w}(r,\nu)\right)^\beta_\alpha
+ \sum_{r\geq 0} a_{\beta,m-r+w} \left(M^{m-r+w}_{-w}(r,\nu)\right)^\beta_\alpha \right). \nn
\end{eqnarray}
Here in the last equality we applied~\eqref{etaMtransp413} to the first sum.
\end{proof}
\begin{proposition}
The following formula holds true for arbitrary $m,n\geq -1$:
\begin{align}
& \left[L_{m}(\nu),L_{n}(\tilde{\nu})\right] \nn\\
&  =\f12 \sum_{p\in\mathbb{Z}+\f12 \atop r,s\geq 0}  :a_p \left( (-1)^{m-r-p-\f12 }  M^p_{m-r-p}(r,\nu) 
M^{p-m+r}_{m+n-r-s-p}(s,\tilde{\nu}) \right.\nn\\
& \left. \quad +  (-1)^{p-n+s-\f12 }  M^p_{n-s-p}(s,\tilde{\nu}) 
M^{p-n+s}_{m+n-r-s-p}(r,\nu)\right) a^{m+n-r-s-p}: \nn \\
& \quad + \f{1}2  \delta_{m,1}  \delta_{n,-1} \left({\rm tr}\left(\f{1}{4}-\mu^2\right)-l  \nu^2\right)
- \f{1}2  \delta_{m,-1}  \delta_{n,1} \left({\rm tr}\left(\f{1}{4}-\mu^2\right)-l  \tilde{\nu}^2\right). \label{comm_nu_nu'}
\end{align}
\end{proposition}

\begin{proof} For $m,n\ge -1$, we have
\begin{eqnarray}
&& \left[L_{m}(\nu), L_{n}(\tilde{\nu})\right]\nn\\
&=&\f12  \Biggl[\sum_{p,q\in\mathbb{Z}+\f12,r\geq 0 \atop p+q+r=m}
: a_{\beta,p}  (M^p_q(r,\nu))^\beta_\gamma  \eta^{\gamma\epsilon}  a_{\epsilon,q}:\, , L_n(\tilde{\nu}) \Biggr]\notag\\
&=&\f12  \sum_{p,q\in\mathbb{Z}+\f12,r\geq 0 \atop p+q+r=m}
\left( \left[a_{\beta,p}, L_n(\tilde{\nu})\right] (M^p_q(r,\nu))^\beta_\gamma  \eta^{\gamma\epsilon}  a_{\epsilon,q}
+ a_{\beta,p}  (M^p_q(r,\nu))^\beta_\gamma  \eta^{\gamma\epsilon} \left[a_{\epsilon,q},L_n(\tilde{\nu})\right]\right)\notag\\
&=&\f12  \sum_{p,q\in\mathbb{Z}+\f12,r\geq 0 \atop p+q+r=m}
\left((-1)^{p-\f12 }\sum_{s\geq 0} (M^{n-s+p}_{-p}(s,\tilde{\nu}))^\alpha_\beta \, a_{\alpha,n-s+p}  (M^p_q(r,\nu))^\beta_\gamma  \eta^{\gamma\epsilon}  a_{\epsilon,q}\right.\notag\\
&&\left. + a_{\beta,p} \, (M^p_q(r,\nu))^\beta_\gamma  \eta^{\gamma\epsilon}  (-1)^{q-\f12}  \sum_{s\geq 0} (M^{n-s+q}_{-q}(s,\tilde{\nu}))^\alpha_\epsilon   a_{\alpha,n-s+q}\right)\nn\\
&=&\f12  \sum_{p\in\mathbb{Z}+\f12 \atop r,s\geq 0}
\left((-1)^{p-n+s-\f12 }  a_{\alpha,p} \, (M^{p}_{n-s-p}(s,\tilde{\nu}))^\alpha_\beta   (M^{p-n+s}_{m+n-r-s-p}(r,\nu))^\beta_\gamma  \eta^{\gamma\epsilon}  a_{\epsilon,m+n-r-s-p}\right.\notag\\
&& \left. + (-1)^{m-r-p-\f12 } \, a_{\beta,p}  (M^p_{m-r-p}(r,\nu))^\beta_\gamma \, (M_{m+n-r-s-p}^{p-m+r}(s,\tilde{\nu}))^\gamma_\epsilon \, \eta^{\alpha\epsilon} \, a_{\alpha,m+n-r-s-p}\right).\notag
\end{eqnarray}
Note that when $p<m+n-r-s-p$, we have
\[
a_{\alpha,p}  a_{\beta,m+n-r-s-p} = \; :  a_{\alpha,p}  a_{\beta,m+n-r-s-p}:,
\]
and when $p\geq m+n-r-s-p$, we have
\[
a_{\alpha,p}  a_{\beta,m+n-r-s-p} = \; :a_{\alpha,p}  a_{\beta,m+n-r-s-p}: +  (-1)^{p-\f12}  \eta_{\alpha\beta}  \delta_{m+n-r-s,0} ,
\]
so we find
\begin{eqnarray}
&& \left[L_m(\nu),L_n(\tilde{\nu})\right]\nn\\
&=&
\f12  \sum_{p\in\mathbb{Z}+\f12 \atop r,s\geq 0}  : a_p \left((-1)^{m-r-p-\f12 }  M^p_{m-r-p}(r,\nu) 
M^{p-m+r}_{m+n-r-s-p}(s,\tilde{\nu}) \right.\notag\\
&& \left. +  (-1)^{p-n+s-\f12} \, M^p_{n-s-p}(s,\tilde{\nu}) 
M^{p-n+s}_{m+n-r-s-p}(r,\nu)\right)  a^{m+n-r-s-p}  : \notag\\
&& + \f{1}2 \sum_{0\leq r<m}(-1)^{m-r-1}  \sum_{\f12 \leq p\leq m-r-\f12}   {\rm tr} \left(M^p_{m-r-p}(r,\nu) M^{p-m+r}_{-p}(m+n-r,\tilde{\nu})\right) \notag\\
&& +  \f{1}2  \sum_{m< r\leq m+n}(-1)^{m-r}  \sum_{\f12 \leq p \leq r-m-\f12}  {\rm tr} \left(M^{p+m-r}_{-p}(r,\nu)  M^p_{-p-m+r}(m+n-r,\tilde{\nu})\right).\notag
\end{eqnarray}
Using an argument similar to the one given in~\cite{DZ2}, 
namely by using~\eqref{p0924} and~\eqref{p20411}, we can calculate the last two terms more explicitly and obtain
\begin{align}
\left[L_m(\nu),L_n(\tilde{\nu})\right]=& \, \f12 \sum_{p\in\mathbb{Z}+\f12 \atop r,s\geq 0}
:a_p \left((-1)^{m-r-p-\f12}  M^p_{m-r-p}(r,\nu)  M^{p-m+r}_{m+n-r-s-p}(s,\tilde{\nu})\right. \notag\\
&\left.+ (-1)^{p-n+s-\f12}  M^p_{n-s-p}(s,\tilde{\nu})  M^{p-n+s}_{m+n-r-s-p}(r,\nu)\right)a^{m+n-r-s-p}: \notag\\
&+ \f{1}2  \delta_{m+n,0}  (-1)^{m-1}  \sum_{\f12 \leq p\leq m-\f12} {\rm tr} \left(M^p_{m-p}(0,\nu) M^{p-m}_{-p}(0,\tilde{\nu})\right) \notag\\
&+ \f{1}2  \delta_{m+n,0}  (-1)^{m}  \sum_{\f12 \leq p \leq -m-\f12} {\rm tr} \left(M^{p+m}_{-p}(0,\nu) M^p_{-p-m}(0,\tilde{\nu})\right) \notag\\
= & \, \f12  \sum_{p\in\mathbb{Z}+\f12 \atop r,s\geq 0} :  a_p \left((-1)^{m-r-p-\f12} M^p_{m-r-p}(r,\nu) M^{p-m+r}_{m+n-r-s-p}(s,\tilde{\nu}) \right. \notag\\
& \left. + (-1)^{p-n+s-\f12} M^p_{n-s-p}(s,\tilde{\nu}) M^{p-n+s}_{m+n-r-s-p}(r,\nu)\right)a^{m+n-r-s-p}: \notag\\
&+ \f{1}2  \delta_{m,1}  \delta_{n,-1}  {\rm tr} \left(M^{\f12}_{\f12}(0,\nu)  M^{-\f12}_{-\f12}(0,\tilde{\nu})\right) \notag\\
&- \f{1}2  \delta_{m,-1}  \delta_{n,1}  {\rm tr} \left(M^{-\f12}_{-\f12}(0,\nu)  M^{\f12}_{\f12}(0,\tilde{\nu})\right). \notag
\end{align}
Here the second equality is due to a direct computation.
The proposition then follows from the simple fact that
\begin{equation}
M^{\f12 }_{\f12}(0,\nu) = -\left(\nu^2+\mu^2-\f{1}{4}\right)  I, \quad M^{-\f12 }_{-\f12}(0,\nu) = I.
\end{equation}
\end{proof}

Let us proceed to prove Theorem~\ref{viralikealgebra}.
The next proposition is important.

\begin{proposition}\label{univeral_comm_n=1,l=0,2}
The following equalities hold true for arbitrary $m\geq -1$, $k\geq 0$:
\begin{align}
& \left[L_{m,2k},L_{1,0}\right] =
(m-1)  L_{m+1,2k} +\f{2}{(m+2)(2k-1)}  \sum_{h\ge k} 
 (2h-4k+3) 
\binom{2h}{2k-2}  L_{m+1,2h}, \label{comm_m_n=1,l=0} \\
& [L_{m,2k},L_{1,2}] = - \frac2{m+2}  \sum_{h\ge k+1}  \binom{ 2h}{2k}  L_{m+1,2h}
+ \delta_{m,-1}  \delta_{k,0}  \frac{l}2 . \label{comm_m_n=1,l=2}
\end{align}
\end{proposition}
\begin{proof}
We first prove the proposition for the case $R=0$. In this case,
for convenience we omit the argument~$r$ in $N^p_q(r,\nu)$ and $M^p_q(r,\nu)$,
and write them as $N^p_q(\nu)$ and $M^p_q(\nu)$ respectively.
The explicit expressions of these matrix-valued functions are given by
\begin{align}
N^p_{m-p}(\nu) 
& = (-1)^{p-m-\f12}  (m+1)!  \binom{\nu+\mu+p}{m+1} \nn
\end{align}
and
\[
M^p_{m-p}(\nu) = (-1)^{p-m-\f12} \, \f{(m+1)!}{2}  \left(
\binom{\nu+\mu+p}{m+1}
+\binom{-\nu+\mu+p}{m+1}
\right).
\]
Here, $p\in\mathbb{Z}+\f12$, $m\geq -1$.
Substituting these expressions into equations~\eqref{comm_nu_nu'}, we find that
\begin{align}
& \left[L_m(\nu),  L_n(\tilde{\nu})\right] = -  \f{(m+1)!  (n+1)!}{8} 
\sum_{p\in\mathbb{Z}+\f12}(-1)^{p-m-n-\f12} \nn\\
& :  a_p  \left\{\left(\binom{\nu+\mu+p}{m+1}
+\binom{-\nu+\mu+p}{m+1} \right) \left(\binom{\tilde{\nu}+\mu+p-m}{n+1}
+\binom{-\tilde{\nu}+\mu+p-m}{ n+1} \right) \right. \nn\\
& \left. - \left(\binom{\tilde{\nu}+\mu+p}{ n+1} + \binom{-\tilde{\nu}+\mu+p}{n+1} \right)
\left(\binom{\nu+\mu+p-n}{ m+1} + \binom{-\nu+\mu+p-n}{m+1} \right)\right\} \, a^{m+n-p}  : \nn \\
& +  \f{1}2  \delta_{m,1}  \delta_{n,-1}  \left(-l \nu^2 + {\rm tr}\left(\f{1}{4}-\mu^2\right)\right)
- \f{1}2 \delta_{m,-1} \delta_{n,1} \left(-l  \tilde{\nu}^2 + {\rm tr}\left(\f{1}{4}-\mu^2\right)\right), \label{comm_nu_nu'_explicit}
\end{align}
where $m,n\geq-1$.
On the other hand, for $m+n\geq-1$ we have
\begin{align}
&L_{m+n}(\nu) \nn\\
& \quad = \f{(m+n+1)!}{4}  \sum_{p\in\mathbb{Z}+\f12 }  (-1)^{p-m-n-\f12} 
:  a_p  \left(\binom{\nu+\mu+p}{m+n+1}+\binom{-\nu+\mu+p}{m+n+1} \right)  a^{m+n-p}  : \nn\\
& \qquad +  \f{\delta_{m+n,0}}{4} \, {\rm tr} \left(\f{1}{4}-\mu^2\right) . \label{L_m+n} 
\end{align}

\begin{lemma}\label{coefficients} For every $m\geq -1$, the following elementary identity holds true:
\begin{equation}\label{elementarybinompoly}
\binom{\nu+x}{ m+1} = \sum_{a=0}^{m+1}\sum_{b=a}^{m+1} \f{(-1)^{b-a}}{b!}\left[\begin{array}{c}b\\ a\\ \end{array}\right] \binom{x}{m+1-b} \nu^a.
\end{equation}
Here, $x$ is an indeterminate, and we recall that
$\left[\begin{array}{cc}b\\a\end{array}\right]$ denotes the unsigned Stirling number of the first kind (see~\eqref{stirling424}).
\end{lemma}
\begin{proof}
Follows from the definition of unsigned Stirling numbers of the first kind.
\end{proof}

Taking $n=1$ in the expressions~\eqref{comm_nu_nu'_explicit} and~\eqref{L_m+n}, and employing Lemma~\ref{coefficients} to compare
the resulting expressions, we find that in order to show the identities~\eqref{comm_m_n=1,l=0}
 and~\eqref{comm_m_n=1,l=2}, it is equivalent to show that the following two elementary identities are valid for 
arbitrary $m\geq -1$, $k\geq 0$:
\begin{align}
& - 2!  (m+1)! \sum_{k_1=2k}^{m+1} 
\f{(-1)^{k_1}}{k_1!} \left[\begin{array}{cc} k_1\\2k\end{array}\right]
\left(\binom{x}{m+1-k_1}  \binom{x-m}{2} - \binom{x}{2}  \binom{x-1}{m+1-k_1}\right)     \nn\\
& \qquad =  (m+2)!  (m-1)  \sum_{b=2k}^{m+2}  \f{(-1)^b}{b!} \left[\begin{array}{cc} b\\2k\end{array}\right] \binom{x}{m+2-b}\nn\\
& \qquad \quad +  2  \f{(m+1)!}{2k-1}  \sum_{h=k}^{[\f{m+2}{2}]}  (2h-4k+3)  \binom{2h}{2k-2} 
\sum_{b=2h}^{m+2}  \f{(-1)^b}{b!} \left[\begin{array}{cc} b\\2h\end{array}\right]
\binom{x}{m+2-b}, \label{id1}\\
& - \frac{2!  (m+1)!}2  \sum_{k_1=2k}^{m+1} 
\f{(-1)^{k_1}}{k_1!} \left[\begin{array}{c} k_1\\2k\end{array}\right]
\left( \binom{x}{m+1-k_1} - \binom{x-1}{m+1-k_1}\right)\notag\\
&\qquad =  -  2  \f{(m+2)!}{(m+2)} \, \sum_{h=k+1}^{[\f{m+2}{2}]}  \binom{2h}{2k} 
\sum_{b=2h}^{m+2}  \f{(-1)^b}{b!}\left[\begin{array}{cc} b\\2h\end{array}\right]
\binom{x}{m+2-b}. \label{id2}
\end{align}
Let us prove these two identities.
Denote by~$a_{m,k}(x)$ the following polynomial in~$x$:
\begin{equation}\label{defamk425}
a_{m,k}(x) = \sum_{k_1=k}^{m+1}  \f{(-1)^{k_1-k}}{k_1!}
\left[\begin{array}{c} k_1\\ k \end{array} \right]
\binom{x}{ m+1-k_1}.
\end{equation}
Introduce the generating function
\beq\label{generatingaA423425}
A_{k,x}(y) := \sum_{m=-1}^\infty a_{m,k}(x)  y^{m+1} = \f{1}{k!}  (1+y)^x  \log^{k}(1+y).
\eeq
In terms of this generating function, the identities~\eqref{id1} and~\eqref{id2} are equivalent to
\begin{align}\label{generating_id1_R0}
&-  (x^2+x)  A_{2k,x}(y) - y^2  A_{2k,x}''(y) + 2  x  y  A_{2k,x}'(y)+(x^2-x)  A_{2k,x-1}(y)\notag\\
&=y A_{2k,x}''(y) - 2  A_{2k,x}'(y) + \f{2}{2k-1}  \sum_{h\geq k}  (2h-4k+3)  \binom{2h}{2k-2}  \f{A_{2h,x}(y)}{y},
\end{align}
and
\begin{equation}\label{generating_id2_R0}
-\left(A_{2k,x}(y)-A_{2k,x-1}(y)\right) = -2\sum_{h\geq k+1} \binom{2h}{2k} \f{A_{2h,x}(y)}{y},
\end{equation}
respectively.
Note that both sides of~\eqref{generating_id1_R0} are equal to
\begin{align}
& \f{(1+ y)^{x-2}}{(2k)!}   \log^{2k-2}(1+ y) \nn\\
& \times  \left(2  (1 - 2 k)  k  y^2 + 2  k  y  (2 x + y)  \log(1+ y) +
   x  (xy-3y-2)  \log^2(1+ y)\right),\nn
\end{align}
and that both sides of~\eqref{generating_id2_R0} are equal to
\[ -  \f{y}{(2k)!}  (1+y)^{x-1}  \log^{2k}(1+y), \]
thus we complete the proof of the identities~\eqref{id1} and~\eqref{id2}, and the identities~\eqref{comm_m_n=1,l=0} and~\eqref{comm_m_n=1,l=2}
hold true for $R=0$.

We continue to consider the general case.
As a generalization of Lemma~\ref{coefficients}, we have the following lemma.
\begin{lemma}\label{coefficients_R} For arbitrary $m\geq -1,r\geq 0$ we have
\begin{equation}
\left(e^{R\pa_\nu}\right)_r  \left(\binom{\nu+x+r}{m+1}\right)
= \sum_{k=0}^{m+1}\sum_{a=k}^{m+1} 
\binom{a}{k}  \sum_{b=a}^{m+1}  \f{(-1)^{b-a}}{b!}  \left[\begin{array}{c}b\\ a\\ \end{array}\right] \binom{x}{m+1-b}  \nu^k  \left(R^{a-k}\right)_r.
\end{equation}
\end{lemma}
The proof is similar to that of Lemma~\ref{coefficients}, so we omit the details.

Using formulas~\eqref{comm_nu_nu'}, \eqref{M} and~\eqref{N}, we find that, for $m,n\geq -1$,
\begin{align}
&\left[L_m(\nu),L_n(\tilde{\nu})\right]\nn\\
=& - \frac{(m+1)!  (n+1)!}{2} \sum_{p\in\mathbb{Z}+\f12}  \sum_{t\geq 0} \sum_{r+s=t}  (-1)^{p+m+n+t-\f12 }\notag\\
&  :  a_p \left\{\overline{\left(e^{R\pa_\nu}\right)_r  \binom{\nu+\mu+p+r}{m+1}}
 \overline{\left(e^{R\pa_{\tilde{\nu}}}\right)_s  \binom{\tilde{\nu}+\mu+p-m+r+s}{n+1}}\right.\notag\\
& \left.\quad\quad  -  \overline{\left[e^{R\pa_{\tilde{\nu}}}\right]_s  \binom{\tilde{\nu}+\mu+p+s}{n+1}} 
\overline{\left(e^{R\pa_\nu}\right)_r  \binom{\nu+\mu+p-n+r+s}{m+1}}\right\}  a^{m+n+t-p}  :\notag\\
&  +  \f{1}2  \delta_{m,1}  \delta_{n,-1}  \left(- l \nu^2+ {\rm tr} \left(\f{1}{4}-\mu^2\right)\right)
-  \f{1}2  \delta_{m,-1} \delta_{n,1}  \left(-l \tilde{\nu}^2+ {\rm tr} \left(\f{1}{4}-\mu^2\right)\right). \label{comm_nu_nu'_explicit_R}
\end{align}
Here, the long bars denote taking the even degree terms of~$\nu$ and of~$\tilde{\nu}$.

On the other hand, from~\eqref{Virasoro-like_op_with_nu} and~\eqref{N} we find that, for $m+n\geq-1$,
\begin{align}\label{L_m+n_R421}
& L_{m+n}(\nu) =  \frac{(m+n+1)!}{2}  \sum_{p\in\mathbb{Z}+\frac{1}{2}} \sum_{t\geq 0}  (-1)^{p+m+n+t-\frac{1}{2}} 
:  a_p  \overline{\left(e^{R\pa_\nu}\right)_{t} \binom{\nu+\mu+p+t}{m+n+1}}  a^{m+n+t-p}  : \nn\\
& \qquad\qquad\quad  +  \frac{\delta_{m+n,0}}{4}  {\rm tr} \left(\frac{1}{4}-\mu^2\right).
\end{align}

By looking at the degree $(2k,0)$-term of $\nu,\tilde{\nu}$ in~\eqref{comm_nu_nu'_explicit_R} with $n=1$
as well as the degree $2h$-terms, $h\geq k$, in~\eqref{L_m+n_R421} with $n=1$, and by using Lemma~\ref{coefficients_R},
we obtain the following equivalent form of the identity~\eqref{comm_m_n=1,l=0}: for arbitrary $m\geq -1,t\geq 0$,
\begin{align}\label{id1_R}
& - 2! (m+1)! \sum_{r+s=t} \Bigg( \sum_{k_1=2k}^{m+1} \binom{k_1}{2k} \sum_{k_2\geq k_1}
\f{(-1)^{k_2-k_1}}{k_2!}\left[\begin{array}{cc} k_2\\k_1\end{array}\right]
\left(\begin{array}{cc} x\\m+1-k_2\end{array}\right) \left(R^{k_1-2k}\right)_r \notag\\
&\quad\quad\quad\quad\quad\quad\times\sum_{l_1=0}^2 \sum_{l_2\geq l_1} \f{(-1)^{l_2-l_1}}{l_2!}\left[\begin{array}{cc} l_2\\l_1\end{array}\right]
\left(\begin{array}{cc} x-m+r\\2-l_2\end{array}\right) \left(R^{l_1}\right)_s \notag\\
&\quad\quad\quad\quad\quad\quad-\sum_{l_1=0}^2 \sum_{l_2\geq l_1} \f{(-1)^{l_2-l_1}}{l_2!}\left[\begin{array}{cc} l_2\\l_1\end{array}\right]
\left(\begin{array}{cc} x\\2-l_2\end{array}\right) \left(R^{l_1}\right)_s\notag\\
&\quad\quad\quad\quad\quad\quad\times \sum_{k_1=2k}^{m+1} \binom{ k_1 }{2k} \sum_{k_2\geq k_1} \f{(-1)^{k_2-k_1}}{k_2!}\left[\begin{array}{cc} k_2\\k_1\end{array}\right]
\left(\begin{array}{cc} x-1+s\\m+1-k_2\end{array}\right) \left(R^{k_1-2k}\right)_r\Bigg)\notag\\
= & (m+2)! (m-1) \sum_{k_1=2k}^{m+2} \binom{k_1}{2k}
\sum_{k_2\geq k_1}
\f{(-1)^{k_2-k_1}}{k_2!}\left[\begin{array}{cc} k_2\\k_1\end{array}\right]
\left(\begin{array}{cc} x\\m+2-k_2\end{array}\right) \left(R^{k_1-2k}\right)_t\notag\\
& + \f{2(m+2)!}{(m+2)(2k-1)} \sum_{h=k}^{[\f{m+2}{2}]} (2h-4k+3)
\binom{2h}{2k-2}
\sum_{k_1=2h}^{m+2} \binom{k_1}{2h}\notag\\
&\quad\quad\quad\quad\quad\quad\quad\quad\quad\quad\times\sum_{k_2\geq k_1}\f{(-1)^{k_2}}{k_2!}
\left[\begin{array}{cc} k_2\\k_1\end{array}\right]
\binom{x}{m+2-k_2} \left(R^{k_1-2h}\right)_t.
\end{align}

Introduce the generating function
\begin{equation}
A_{x,k,t}(y) := \sum_{m=-1}^\infty \, \sum_{k_1=k}^{m+1} \binom{k_1}{k} \sum_{k_2\geq k_1}
\f{(-1)^{k_2-k_1}}{k_2!} \left[\begin{array}{cc} k_2\\k_1\end{array}\right]
\binom{x}{m+1-k_2} \left(R^{k_1-k}\right)_t y^{m+1}.
\end{equation}
We claim that for $t,k\geq 0$, $A_{x,k,t}(y)$ has the expression
\begin{equation}\label{generating function with R}
A_{x,k,t}(y) = \f{(1+y)^x}{k!} \log^{k}(1+y) \left((1+y)^R\right)_t.
\end{equation}
Indeed,
\begin{align}
{\rm RHS~of~}\eqref{generating function with R} & = \f{(1+y)^x}{k!} \sum_{\beta=0}^\infty
\f{1}{\beta!} \log^{k+\beta}(1+y) \left(R^\beta\right)_t \nn\\
&= \f{1}{k!} \sum_{\alpha=0}^\infty \binom{x}{\alpha} y^\alpha 
\sum_{\beta=0}^\infty \f{(k+\beta)!}{\beta!} \sum_{\tilde k=k+\beta}^\infty
\f{(-1)^{\tilde k-k-\beta}}{\tilde k!} \left[\begin{array}{c} \tilde k\\ k+\beta \end{array}\right] y^{\tilde k} \left(R^\beta\right)_t \nn\\
&= \sum_{m=-1}^\infty \sum_{\beta=0}^\infty \sum_{\tilde k=k+\beta}^\infty \f{(k+\beta)!}{k! \beta!} \binom{x}{m+1-\tilde k}
 \f{(-1)^{\tilde k-k-\beta}}{\tilde k!} \left[\begin{array}{c} \tilde k\\ k+\beta \end{array}\right] y^{m+1} \left(R^\beta\right)_t\notag\\
&= \sum_{m=-1}^\infty\sum_{k_1=k}^\infty\sum_{\tilde k=k_1}^\infty \binom{k_1}{k} \binom{ x }{m+1-\tilde k} \,
\f{(-1)^{\tilde k-k_1}}{\tilde k!} \left[\begin{array}{c} \tilde k \\ k_1 \end{array}\right] y^{m+1} \left(R^{k_1-k}\right)_t \notag\\
&= {\rm LHS~of~}\eqref{generating function with R}.
\end{align}
In terms of this generating function,
the identity~\eqref{id1_R} can be represented as the following identity for arbitrary $t, k\geq 0$:
\begin{align}\label{generating_id1_R_equiva}
&\quad - y^2 A_{x,2k,t}''(y) + 2 x y A_{x,2k,t}'(y) + 2 y \sum_{r+s=t} A_{x,2k,r}'(y) R_s\notag\\
&\quad - x (x+1) A_{x,2k,t}(y)-(2x+1) \sum_{r+s=t} A_{x,2k,r}(y) R_s - \sum_{r+s=t} A_{x,2k,r}(y) \left(R^2\right)_s\notag\\
&\quad + x (x-1) A_{x,2k,t}(y)+(2x-1) \sum_{r+s=t} A_{x-1,2k,r}(y) R_s + \sum_{r+s=t} A_{x-1,2k,r}(y) \left(R^2\right)_s\notag\\
&= y A_{x,2k,t}''(y) - 2 A_{x,2k,t}'(y) + \f{2}{2k-1} \sum_{h\geq k} (2h-4k+3) \binom{2h}{2k-2} \f{A_{x,2h,t}(y)}{y}.
\end{align}

To prove the identity~\eqref{generating_id1_R_equiva} we define
\begin{align}
&w_1 := \f{1}{(2k)!} (1+y)^{x-1} \log^{2k}(1+y) \left(R (1+y)^R\right)_t,\\
&w_2 := \f{1}{(2k)!} (1+y)^{x-2} \log^{2k}(1+y) \left(R^2 (1+y)^R\right)_t.
\end{align}
Then
\begin{align}
&A_{x,2k,t}'(y) = \left(\f{x}{1+y}+\f{2k}{(1+y)\log(1+y)}\right) \, A_{x,2k,t}(y) + w_1,\\
&A_{x,2k,t}''(y) = B(x,2k,t,y) A_{x,2k,t}(y) + \left(\f{2x-1}{1+y}+\f{4k}{(1+y)\log(1+y)}\right)w_1 + w_2
\end{align}
for some function $B(x,2k,t,y)$. Noticing that
\begin{align}
&\sum_{r+s=t}A_{x,2k,r}(y) R_s = (1+y) w_1,\notag\\
&\sum_{r+s=t}A_{x,2k,r}'(y) R_s = x w_1 + \f{2k}{\log(1+y)} w_1 + (1+y) w_2,\notag\\
&\sum_{r+s=t}A_{x,2k,r}(y) \left(R^2\right)_s = (1+y)^2 w_2, \notag
\end{align}
and comparing~\eqref{generating_id1_R_equiva} with the already proved identity~\eqref{generating_id1_R0}, we find that it suffices to show
\begin{align}\label{eq1_R}
&\quad - y^2 \left(\f{2x-1}{1+y}w_1+\f{4k}{(1+y)\log(1+y)}w_1+w_2\right) + 2 x y w_1\notag\\
&\quad + 2 \, y \left(x w_1+\f{2k}{\log(1+y)} w_1+(1+y) w_2\right) - (2x+1) (1+y) w_1 - (1+y)^2 w_2\notag\\
&\quad + (2x-1) w_1 + (1+y) w_2\notag\\
&= \f{y (2x-1)}{1+y} w_1 + \f{4 k y}{(1+y)\log(1+y)} w_1 + y w_2.
\end{align}
The validity of this equality can be verified easily. Hence
we complete the proof of the identity~\eqref{comm_m_n=1,l=0}.

Similarly, by looking at the degree $(2k,2)$-term of $\nu,\tilde{\nu}$ in~\eqref{comm_nu_nu'_explicit_R} with $n=1$ and the degree $2h$-terms,
$h\geq k$, in~\eqref{L_m+n_R421} with $n=1$, and by using Lemma~\ref{coefficients_R}, we obtain the following equivalent form
of the identity~\eqref{comm_m_n=1,l=2}: for arbitrary $m\geq -1,t\geq 0$,
\begin{align}\label{id2_R}
& - \, \frac{2! (m+1)!}2 \left\{ \sum_{k_1=2k}^{m+1} \binom{k_1}{2k} \sum_{k_2\geq k_1}
\f{(-1)^{k_2-k_1}}{k_2!} \left[\begin{array}{cc} k_2\\k_1\end{array}\right]
\binom{ x }{m+1-k_2} \left(R^{k_1-2k}\right)_t  \right.\notag\\
& \left. \quad\quad\quad\quad\quad - \sum_{k_1=2k}^{m+1} 
\binom{k_1}{2k} \sum_{k_2\geq k_1} \f{(-1)^{k_2-k_1}}{k_2!}\left[\begin{array}{cc} k_2\\k_1\end{array}\right] \,
\binom{x-1+s}{m+1-k_2} \left(R^{k_1-2k}\right)_t \right\}\notag\\
= & \frac{-2(m+2)!}{m+2} \sum_{h\geq k+1} \,
\binom{2h}{2k} \sum_{k_1\geq2h} \binom{k_1}{2h} \sum_{k_2\geq k_1} \f{(-1)^{k_2}}{k_2!} \left[\begin{array}{cc} k_2\\k_1\end{array}\right]
\binom{x}{m+2-k_2} \left(R^{k_1-2h}\right)_t.
\end{align}
In terms of the generating function $A_{x,k,t}(y)$, the identity~\eqref{id2_R} is equivalent to
\begin{equation}\label{generating_id2_R}
\frac{1}{2} \left(A_{x,2k,t}(y)-A_{x-1,2k,t}(y) \right) = \sum_{h\geq k+1} \binom{2h}{2k} \f{A_{x,2h,t}(y)}{y}.
\end{equation}
Comparing~\eqref{generating_id2_R} with the already proved identity~\eqref{generating_id2_R0},
one can obtain the validity of~\eqref{generating_id2_R}. This proves the identity~\eqref{comm_m_n=1,l=2}.
The proposition is proved.
\end{proof}

We note that Lemma~\ref{linearindep923} was already implicitly given in the proof of Proposition~\ref{univeral_comm_n=1,l=0,2}. Let us make it 
more explicitly.

\begin{proof}[Proof of Lemma~\ref{linearindep923}]
By using the expression~\eqref{L_m+n_R421} and Lemma~\ref{coefficients_R}.
\end{proof}

We are ready to prove Theorem~\ref{viralikealgebra}.

\begin{proof}[Proof of Theorem~\ref{viralikealgebra}]
For simplification of notations, define $L_{m,2k}=0$ if $k> [(m+1)/2]$. We are to show that
$[L_{m,2k}, L_{n,2\ell}] \in {\rm Vira}_{\rm like}$ for all $m,n\geq-1$.

Firstly, from~\eqref{Lmn0407} we know that $[L_{m,0},L_{n,0}]=(m-n) L_{m+n,0} \in {\rm Vira}_{\rm like}$, $\forall \, m,n\geq-1$. (One can also verify this using~\eqref{comm_nu_nu'}.)

Secondly, let us prove by induction
that $[L_{m,2k},L_{n,2\ell}] \in {\rm Vira}_{\rm like}$ for any fixed $m\geq1$ and for all $n\geq 1$, $k\geq1$, $\ell\geq 0$.
From the identity~\eqref{comm_m_n=1,l=0} we know that this is true for $m = 1$.
Suppose that $[L_{m,2k},L_{n,2\ell}] \in {\rm Vira}_{\rm like}$ 
for $1\leq m \leq M$~$(M\geq1)$ and arbitrary $n\geq 1$, $k\geq1$, $\ell\geq 0$.
Consider $m=M+1$. Using~\eqref{comm_m_n=1,l=2}, we find that
\[L_{M+1,2[M/2]+2}=c \left[L_{1,2},L_{M,2[M/2]}\right]\]
for a certain constant~$c$. So 
\begin{align*}
\left[L_{M+1,2[M/2]+2]},L_{n,\ell}\right] &= c \left[\left[L_{1,2},L_{M,2[M/2]}\right],L_{n,2\ell}\right] \\
&=
c \left[[L_{1,2},L_{n,2\ell}],L_{M,2[M/2]}\right] + c \left[L_{1,2}, \left[L_{M,2[M/2]},L_{n,2\ell}\right]\right],
\end{align*}
which
belongs to ${\rm Vira}_{\rm like}$.
In a similar way, by using~\eqref{comm_m_n=1,l=2} we find
that for $k=[M/2], [M/2]-1,\dots,1$ we have $[L_{M+1,2k},L_{n,2\ell}]\in {\rm Vira}_{\rm like}$.
Hence we have proved the statement for any fixed $m\geq1$
and for all $n\geq 1$, $k\geq1$, $\ell\geq 0$.

It suffices to show the following remaining cases: 
\begin{enumerate}
\item[A.] $(m=-1, k=0, n\geq 1, \ell\geq1)$; 
\item[B.] $(m=0,k=0,n\geq1,\ell\geq1)$.
\end{enumerate}
Namely, we consider now either $(m=-1,k=0)$ or $(m=0,k=0)$.
If $n=1$ then the statement is true due to the identity~\eqref{comm_m_n=1,l=2}.
Suppose the statement is true for $n\le N$ and arbitrary $\ell\ge1$, then for $n=N+1$,
we have
\[L_{N+1,2[(N+2)/2]}=c [L_{1,2},L_{N,2[(N+1)/2]}]\]
for a certain constant~$c$.
So 
\begin{align*}
[L_{m,k},L_{N+1,2[(N+2)/2]}]&= c [L_{m,k},[L_{1,2},L_{N,2[N/2]}]] \\
&=
c [[L_{m,k},L_{1,2}],L_{N,2[N/2]}]+ c [L_{1,2},[L_{m,k},L_{N,2[N/2]}]] \in {\rm Vira}_{\rm like}.
\end{align*}

Hence we completed the proof of the existence part of the theorem. 

In view of Lemma~\ref{linearindep923}, the constants 
$c_{m,2k,n,2\ell,2h}$ $(m,n\geq -1$, $0\leq k \leq[(m+1)/2]$, $0\leq \ell \leq[(n+1)/2]$, 
$0\leq h\leq [(m+n+1)/2])$ if exist must be unique.
From Proposition~\ref{univeral_comm_n=1,l=0,2} we know that the part of the constants 
$c_{m,2k,1,2\ell,2h}$ $(m\geq -1$, $0\leq k \leq[(m+1)/2]$, $0\leq \ell \leq 1$, 
$0\leq h\leq [(m+2)/2])$ are independent of the Frobenius manifold. Then from the 
above induction procedure we know that all the constants are independent of the Frobenius manifold.
The induction procedure also implies the vanishing of $c_{m,2k,n,2\ell,2h}$ whenever $h<k+\ell$.

The theorem is proved.
\end{proof}

The following proposition, as mentioned in the Introduction, 
gives a description of the essential structure constants of~${\rm Vira}_{\rm half}$.

\begin{proposition}\label{corollary10425} The essential structure constants of the Virasoro-like algebra of a
Frobenius manifold satisfy the relation
\begin{align}
& \sum_{k=0}^{[(m+1)/2]} \sum_{\ell=0}^{[(n+1)/2]} \nu^{2k} \tilde{\nu}^{2\ell} \sum_{h = k+\ell}^{[(m+n+1)/2]} c_{m,2k,n,2\ell,2h} \,
\sum_{b=2h}^{m+n+1}\frac{(-1)^b}{b!} \left[\begin{array}{c} b\\ 2h\end{array}\right] \binom{x}{m+n+1-b} \nn\\
&  = - \f{(m+1)! (n+1)!}{4 (m+n+1)!} \nn\\
& \qquad \biggl(\left(\binom{\nu+x}{m+1}
+\binom{-\nu+x}{m+1} \right)\left( \binom{\tilde{\nu}+x-m}{n+1}
+\binom{-\tilde{\nu}+x-m}{ n+1} \right) \nn\\
& \qquad\quad - \left(\binom{\tilde{\nu}+x}{ n+1 }  + \binom{-\tilde{\nu}+x}{n+1} \right)
\left(\binom{\nu+x-n}{ m+1}
+\binom{-\nu+x-n}{ m+1} \right)\biggr),\label{comm_nu_nu'_explicitoncorele424}
\end{align}
for $m,n\geq-1$.
Moreover, the following formulae hold true:
\begin{align}
& c_{m,0,n,0,2h} = (m-n) \delta_{h,0}, \quad 0\leq h\leq [(m+n+1)/2], \label{cm0n0h424}\\
& c_{-1,0,n,2\ell,2h} = - (n+1) \delta_{h,\ell}, \quad \ell\leq h \leq[n/2], \label{cm10n2l2l424}\\
& c_{0,0,n,2\ell,2h} = - n \delta_{h,\ell}, \quad \ell\leq h\leq [(n+1)/2], \label{c00n2l2l424}\\
& c_{1,0,n,2\ell,2h} = - 2 \frac{(2h-4\ell+3)}{(n+2)(2\ell-1)} \, \binom{2h}{2\ell-2} - (n-1) \delta_{h,2\ell}, \quad \ell\leq h\leq [(n+2)/2], \label{c10n2l2h424}\\
& c_{2k-1,2k,n,2\ell,2h}  =  \frac{(2k)!}{\binom{n+2k}{n+1}}  \left\{\begin{array}{c} 2h - 2\ell \\ 2k-1 \end{array}\right\}
 \binom{2h}{2\ell}, \quad k\ge1, \, k+\ell\leq h\leq [(n+2k)/2]. \label{c2km12kn2l2h424}
\end{align}
Here,  $m,n\geq-1$ and $0\leq\ell\leq [(n+1)/2]$.
\end{proposition}
\begin{proof}
We note that the relations \eqref{cm0n0h424}, \eqref{c10n2l2h424} and formula~\eqref{c2km12kn2l2h424} with $k=1$ were already proved above.
Let us prove the remaining relations.
Since the essential structure constants of the Virasoro like algebra are independent of the Frobenius manifold, we can
assume that $M$ is the one-dimensional Frobenius manifold, which has vanishing $\mu$ and $R$.
Applying $\sum_{k,\ell\geq0} \nu^{2k} \, \tilde{\nu}^{2\ell}$ on
both sides of~\eqref{structureviralike422-424} we find
\begin{align}
& - \f{(m+1)! (n+1)!}{8} \sum_{p\in\mathbb{Z}+\f12}(-1)^{p-m-n-\f12} \nn\\
& : a_p \biggl( \left(\binom{\nu+\mu+p}{m+1}
+\binom{-\nu+\mu+p}{m+1} \right) \left( \binom{\tilde{\nu}+\mu+p-m}{n+1}
+\binom{-\tilde{\nu}+\mu+p-m}{ n+1} \right) \nn\\
& - \left(\binom{\tilde{\nu}+\mu+p}{n+1}  + \binom{-\tilde{\nu}+\mu+p}{n+1} \right)
\left(\binom{\nu+\mu+p-n}{ m+1}
+\binom{-\nu+\mu+p-n}{ m+1} \right)\biggr) a^{m+n-p} : \nn \\
& + \f12 \delta_{m,1} \delta_{n,-1} \left(-\nu^2 + \f{1}{4} \right) 
-\f12 \delta_{m,-1} \delta_{n,1} \left(- \tilde{\nu}^2 + \f{1}{4}\right) \nn\\
&  = \sum_{k,\ell\geq0} \nu^{2k} \tilde{\nu}^{2\ell} 
\sum_{h = k+\ell}^{[(m+n+1)/2]} c_{m,2k,n,2\ell,2h} \f{(m+n+1)!}{4} \sum_{p\in\mathbb{Z}+\f12} (-1)^{p-m-n-\f12}  \nn\\
& \quad : a_p {\rm Coef} \left(\left(\binom{z +\mu+p}{m+n+1}
+\binom{-z+\mu+p}{m+n+1} \right), z^{2h}\right) a^{m+n-p} : \nn\\
& \quad +  c_{m,0,n,0,0} \f{\delta_{m+n,0}}{16} - \delta_{m,1} \delta_{n,-1}  \frac{1}2 \nu^2 
+ \delta_{m,-1} \delta_{n,1} \frac{1}2 \tilde{\nu}^2 , \nn
\end{align}
which implies~\eqref{comm_nu_nu'_explicitoncorele424}. Here we used $c_{1,0,-1,0,0}=-c_{-1,0,1,0,0}=2$.

Let us continue to show \eqref{cm10n2l2l424}, \eqref{c00n2l2l424} and~\eqref{c2km12kn2l2h424}. Taking $m=-1$ and
looking at the coefficient of~$\nu^0$ on both sides of~\eqref{comm_nu_nu'_explicitoncorele424} we have
\begin{align}
& \sum_{\ell=0}^{[(n+1)/2]} \tilde{\nu}^{2\ell} \sum_{h = \ell}^{[n/2]} c_{-1,0,n,2\ell,2h}
\sum_{b=2h}^{n}\frac{(-1)^b}{b!} \left[\begin{array}{c} b\\ 2h\end{array}\right] \binom{x}{n-b} \nn\\
& = - \f{n+1}{2} \biggl( \binom{\tilde{\nu}+x+1}{n+1}
 + \binom{-\tilde{\nu}+x+1}{n+1}  - \binom{\tilde{\nu}+x}{n+1} - \binom{-\tilde{\nu}+x}{n+1}
\biggr)\nn\\
& = - \f{n+1}{2} \left(\binom{\tilde{\nu}+x+1}{n} + \binom{-\tilde{\nu}+x+1}{n} \right), \nn
\end{align}
which together with~\eqref{elementarybinompoly} gives~\eqref{cm10n2l2l424}.
Similarly, taking $m=0$ and
looking at the coefficient of~$\nu^0$ on both sides of~\eqref{comm_nu_nu'_explicitoncorele424} we have
\begin{align}
& \sum_{\ell=0}^{[(n+1)/2]} \tilde{\nu}^{2\ell} \sum_{h = k+\ell}^{[(n+1)/2]} c_{0,0,n,2\ell,2h}
\sum_{b=2h}^{n+1}\frac{(-1)^b}{b!} \left[\begin{array}{c} b\\ 2h\end{array}\right] \, \binom{x}{n+1-b} \nn\\
&  = - \f{1}{4} \biggl(\left( 2x \right)\left( \binom{\tilde{\nu}+x}{n+1}
+\binom{-\tilde{\nu}+x}{ n+1} \right) - \left(\binom{\tilde{\nu}+x}{ n+1 }  + \binom{-\tilde{\nu}+x}{n+1} \right)
\left(2x-2n \right)\biggr) \nn\\
& = - \f{n}{2} \left( \binom{\tilde{\nu}+x}{n+1}
+\binom{-\tilde{\nu}+x}{ n+1} \right), \nn
\end{align}
which proves~\eqref{c00n2l2l424}. Finally, taking $m=2k-1$ $(k\ge1)$ and looking at the coefficient of~$\nu^{2k}$
on both sides of~\eqref{comm_nu_nu'_explicitoncorele424}, we have, for $n\ge-1$, that
\begin{align}
& \sum_{\ell=0}^{[(n+1)/2]} \tilde{\nu}^{2\ell} \sum_{h = k+\ell}^{[(2k+n)/2]} c_{2k-1,2k,n,2\ell,2h}
\sum_{b=2h}^{2k+n}\frac{(-1)^b}{b!} \left[\begin{array}{c} b\\ 2h\end{array}\right] \binom{x}{2k+n-b} \nn\\
& = - \f{(2k)! (n+1)!}{2 (2k+n)!}  \left(\binom{\tilde{\nu}+x-(2k-1)}{n+1}
+ \binom{-\tilde{\nu}+x-(2k-1)}{n+1}  \right. \nn\\
& \left. \qquad\qquad\qquad - \binom{\tilde{\nu}+x}{n+1}  - \binom{-\tilde{\nu}+x}{n+1} \right). \label{c2km12kn2l2hequivid425}
\end{align}
From the uniqueness of the essential structure constants it suffices to verify
that if we define $c_{2k-1,2k,n,2\ell,2h}$ by~\eqref{c2km12kn2l2h424},
then~\eqref{c2km12kn2l2hequivid425} becomes an identity, that is,
\begin{align}
& 2 (2k-1)! \sum_{\ell=0}^{[(n+1)/2]} \tilde{\nu}^{2\ell} \sum_{h = k+\ell}^{[(2k+n)/2]} 
\left\{\begin{array}{c} 2h - 2\ell \\ 2k-1 \end{array}\right\} \binom{2h}{2\ell}
\sum_{b=2h}^{2k+n}\frac{(-1)^b}{b!} \left[\begin{array}{c} b\\ 2h\end{array}\right] \binom{x}{2k+n-b} \nn\\
& = - \binom{\tilde{\nu}+x-(2k-1)}{n+1} - \binom{-\tilde{\nu}+x-(2k-1)}{n+1} + \binom{\tilde{\nu}+x}{n+1} + \binom{-\tilde{\nu}+x}{n+1}.  \nn
\end{align}
To show the validity of this identity for arbitrary $n\geq-1$, similarly to the proof of Proposition~\ref{univeral_comm_n=1,l=0,2},
it suffices to show that
\begin{align}
& 2  (2k-1)! \sum_{n=-1}^\infty y^{n+1} \sum_{\ell=0}^{[(n+1)/2]} \tilde{\nu}^{2\ell} \sum_{h = k+\ell}^{[(2k+n)/2]}
\left\{\begin{array}{c} 2h - 2\ell \\ 2k-1 \end{array}\right\} \binom{2h}{2\ell} a_{2k-1+n,2h} \nn\\
& = - (1+y)^{\tilde{\nu}+x-(2k-1)} - (1+y)^{-\tilde{\nu}+x-(2k-1)}  + (1+y)^{\tilde{\nu}+x} + (1+y)^{-\tilde{\nu}+x}, \nn
\end{align}
where we used the notation given by~\eqref{defamk425}.
Now by using~\eqref{generatingaA423425} we find that
the left-hand side can be simplified to
\begin{align}
& 2 (2k-1)! y^{-(2k-1)} (1+y)^x \sum_{\ell\ge0} \tilde{\nu}^{2\ell} \sum_{h \ge k+\ell}
\left\{\begin{array}{c} 2h - 2\ell \\ 2k-1 \end{array}\right\} \frac{\log^{2h}(1+y)}{(2h-2\ell)! \, (2\ell)!} \nn\\
& = 2 (2k-1)! \, y^{-(2k-1)} (1+y)^x \sum_{\ell\ge0} \log^{2\ell}(1+y) \frac{\tilde{\nu}^{2\ell}}{(2\ell)!} \nn\\
& \qquad \times \left(\frac{\left(e^{\log(1+y)}-1\right)^{2k-1}}{(2k-1)!} + \frac{\left(e^{-\log(1+y)}-1\right)^{2k-1}}{(2k-1)!}\right) \nn\\
& = \left(1-(1+y)^{-(2k-1)}\right) (1+y)^x \left((1+y)^{\tilde{\nu}}+(1+y)^{-\tilde{\nu}}\right). \nn
\end{align}
The proposition is proved.
\end{proof}

We note that formula~\eqref{comm_nu_nu'_explicitoncorele424} in the above proposition
could be viewed as a generating formula for the essential structure constants.
The explicit expression~\eqref{c2km12kn2l2h424} will be
used in Section~\ref{section5} for the construction of Virasoro constraints for the Hodge partition function.

\begin{remark}
For the case when $\mu_\alpha$ do not contain half integers,
one can also take $\nu=0$ in~$L_m(\nu)$ for $m\leq -2$. In this case, it is shown in~\cite{DZ-norm} that
\begin{align}
\left[ L_m(0), L_n(0) \right] = (m-n) L_{m+n}(0) +  l \frac{m(m^2-1)}{12} \delta_{m+n,0}.
\end{align}
The central charge for this realization of Virasoro algebra is equal to~$l$.
\end{remark}

We finish this section by proving Proposition~\ref{virahalftwistthm}.

\begin{proof}[Proof of Proposition~\ref{virahalftwistthm}]
Since the essential structure constants of the Virasoro like algebra are independent of a Frobenius manifold, it suffices to prove the proposition for the one-dimensional Frobenius manifold. In this case, it is equivalent to show validity of the following identity
for $\forall\,m,n\geq-1$:
\begin{align}
& \left(\binom{1/2+x}{m+1}
+\binom{-1/2+x}{m+1} \right) \left(\binom{1/2+x-m}{n+1}
+\binom{-1/2+x-m}{n+1} \right) \nn\\
& - \, \left(\binom{1/2+x}{ n+1} + \binom{-1/2+x}{n+1} \right)
\left(\binom{1/2+x-n}{ m+1}+\binom{-1/2+x-n}{ m+1} \right)  \nn\\
& = \; - \, 2\, \f{(m-n)\, (m+n+1)!}{(m+1)! \, (n+1)!} \, \left(\binom{1/2+x}{m+n+1} + \binom{-1/2+x}{m+n+1}\right),\nn
\end{align}
which can be easily verified.
\end{proof}

\section{Genus-zero Virasoro-like constraints}\label{section4}
It was proved in~\cite{DZ-norm} (cf.~also~\cite{DZ2, EHX, LT}) that
 the series~$\mathcal{F}_0(\bt)$ satisfies the following genus zero Virasoro constraints
\begin{equation}\label{viraconstraints}
e^{-\e^{-2}\F_0(\bt)} L_m e^{\e^{-2}\F_0(\bt)} = {\rm O}(1), \quad \epsilon\rightarrow 0, \, m\geq -1.
\end{equation}
Let us prove Theorem~\ref{theoremcor415} by using this result. 

\begin{proof}[Proof of Theorem~\ref{theoremcor415}]
According to~\cite{DZ2},
the theorem is true for $L_{m,0}$, $m\ge -1$.
In order to prove the statement for $L_{m,2k}$ with $m, k\geq1$, we first prove the following lemma.
\begin{lemma}\label{AB}
Let $A,B$ be operators of the form
$$\epsilon^2 \sum a^{\alpha,i;\beta,j} \frac{\p^2}{\p t^{\alpha,i}t^{\beta,j}} + \sum b^{\alpha,i}_{\beta,j}  t^{\beta,j} 
\frac{\p}{\p t^{\alpha,i}} + \frac{1}{\epsilon^2} \sum c_{\alpha,i;\beta,j} t^{\alpha,i}t^{\beta,j} + {\rm const},$$
where $a,b,c$'s are constants.
If as $\epsilon\rightarrow 0$, $$\tau^{-1} A (\tau) = \mathcal{O}(1), \quad \tau^{-1} B (\tau) = \mathcal{O}(1),$$ then
    $$\tau^{-1} [A,B] (\tau) = \mathcal{O}(1).$$
\end{lemma}
\begin{proof} We have
\begin{equation*}
\begin{split}
[A,B] (\tau) &= A \left(\tau \mathcal{O}(1)\right) - B \left(\tau \mathcal{O}(1)\right)\\
&= A (\tau) \mathcal{O}(1) + \tau \mathcal{O}(1) + B (\tau) \mathcal{O}(1) = \tau \mathcal{O}(1).
\end{split}
\end{equation*}
\end{proof}

\begin{lemma} \label{sub-family}
The following identities hold true:
\begin{align}
\label{L2k-12k}
\tau^{-1} L_{2k-1,2k} \tau = \mathcal{O}(1),\\
\label{L2k2k}
\tau^{-1} L_{2k,2k} \tau = \mathcal{O}(1),
\end{align}
as $\epsilon\rightarrow 0$, where $k\geq 1$.
\end{lemma}
\begin{proof}
First let us show that if~\eqref{L2k-12k} holds true for a certain $k\geq1$, then~\eqref{L2k2k} also holds true for this~$k$. Indeed,
By using~\eqref{comm_m_n=1,l=0} we have
$$[L_{1,0},L_{2k-1,2k}] = -\frac{2(2k-1)}{2k+1} L_{2k,2k}.$$
Since $\tau^{-1} L_{1,0} \tau=\mathcal{O}(1)$ and $\tau^{-1} L_{2k-1,2k} \tau = \mathcal{O}(1)$, by using Lemma~\ref{AB} we find that
$\tau^{-1} L_{2k,2k} \tau=\mathcal{O}(1)$.

Let us now show the validity of~\eqref{L2k-12k}.
For $k=1$, the equality~\eqref{L2k-12k} is true due to recursion relations satisfied by $\Omega_{\alpha,p;\beta,q}$ (for the 
details about this verification see e.g.~\cite{DLYZ}).
Suppose the equality~\eqref{L2k-12k} holds true for $k=l$ ($l\geq 1$). Then we know that~\eqref{L2k2k} also holds true for $k=l$.
Using~\eqref{comm_m_n=1,l=2}, we have
$$[L_{1,2},L_{2l,2l}] = (2l+1)L_{2l+1,2l+2}.$$
Since $\tau^{-1} L_{1,2} \tau=\mathcal{O}(1)$, $\tau^{-1} L_{2l,2l} \tau=\mathcal{O}(1)$, using Lemma~\ref{AB}, we find that
\[\tau^{-1} L_{2l+1,2l+2} \tau=\mathcal{O}(1).\] 
The lemma is proved.
\end{proof}

\noindent {\it End of the proof of Theorem~\ref{theoremcor415}.} Note that
the coefficient in front of $L_{m+1,2k}$ in the right-hand side of~\eqref{comm_m_n=1,l=0} is $m-1-\f{2k(2k-3)}{m+2}$, which is \emph{nonzero} for all $m\geq 1$ and $1\leq k\leq [\f{m+1}{2}]$. This fact implies that $L_{m+1,2k}$ can be expressed as linear combinations of
    $[L_{m,2k},L_{1,0}]$ and $L_{m+1,2h}$, $h\geq k+1$.
Then the theorem can be proved by using Lemmas~\ref{AB}, \ref{sub-family} and by induction.
\end{proof}

\noindent We call~\eqref{viralikeconstraints} the \textit{Virasoro-like constraints for the genus 
zero free energy~$\mathcal{F}_0(\bt)$}. As we have mentioned in the Introduction 
the $L_{1,2}$-constraint in~\eqref{viralikeconstraints} was proved 
in~\cite{EHX} and~\cite{LT}, and 
the $L_{2,2}$-constraint in~\eqref{viralikeconstraints} was proved in~\cite{LT}. 

\section{Virasoro constraints for the Hodge partition function} \label{section5}

In this section, we assume that the calibrated Frobenius manifold~$M^l$ under consideration is {\it semisimple}.

\begin{proof}[Proof of Theorem~\ref{theorem30407}] Following~\cite{Zhou},
define
\begin{equation} \label{Vira_deform_s}
L_n^{\rm H} := e^{U} \circ L_{n,0} \circ e^{-U}, \quad n\geq -1,
\end{equation}
where $U := - \sum_{k\geq 1} \frac{\sigma_{2k-1}B_{2k}}{(2k)!}L_{2k-1,2k}$.
The operators $L_{n,0}$, $n\geq -1$, satisfy the Virasoro commutation relations, so do the
 operators $L_n^{\rm H}$. Moreover,
 due to~\eqref{hodgepartitionfunction420} we have
$L_n^{\rm H} Z_H = 0$.

It then suffices to show that the operators $L_n^{\rm H}$ coincide with the ones
proposed in the statement of the theorem.
It follows from Corollary~\ref{corollary10425} that
for $m\geq -1$ and $0\leq k\leq [(m+1)/2]$,
\begin{align}
& \left[L_{m,2k},L_{-1,0}\right] = (m+1) L_{m-1,2k} - \delta_{m,1} \delta_{k,1} \frac{l}{2}, \label{lm2km10421}\\
& \left[L_{m,2k},L_{0,0}\right] = m L_{m,2k}. \label{lm2k00425}
\end{align}
As a particular case of~\eqref{lm2km10421}, we have, for $k\geq1$,
$$ \left[L_{2k-1,2k},L_{-1,0}\right] = - \delta_{k,1} \frac{l}{2}, $$
which leads to the expression~\eqref{lminus1421} for the operator~$L_{-1}^{\rm H}$. 
Let us continue the computation on the explicit expressions of~$L_n^{\rm H}$ with $n\geq0$.
Using equation~\eqref{Vira_deform_s}
and using the fact that the operators $L_{2k-1,2k}$ pairwise commute,
we find
\begin{align}\label{Ln(s)}
L_n^{\rm H}& = \sum_{m=0}^\infty \frac{{\rm ad}_{U}^m}{m!} L_{n,0}\notag\\
& = L_{n,0} + \sum_{m=1}^\infty \frac{1}{m!} \sum_{k_1,...,k_m\geq 1} \prod_{i=1}^m s_{k_1}\cdots s_{k_m} 
{\rm ad}_{L_{2k_m-1,2k_m}} \cdots {\rm ad}_{L_{2k_1-1,2k_1}} L_{n,0}.
\end{align}
Using Corollary~\ref{corollary10425} we find that, for $0\leq\ell\leq[(n+1)/2]$,
\begin{align}
\left[L_{2k-1,2k},L_{n,2\ell}\right] = \frac{(2k)!}{\binom{n+2k}{n+1}}
\sum_{h=k+\ell}^\infty \left\{\begin{array}{c} 2h-2 \ell \\ 2k-1 \end{array}\right\}
 \binom{2h}{2\ell} L_{n+2k-1,2h} - \delta_{k,1} \delta_{n,-1} \delta_{\ell,1} \frac{l}{2}, \label{l2km12kn2lpp425}
\end{align}
where we recall that $\left\{\begin{array}{c} a \\ b\end{array}\right\}$ denotes the Stirling number of
the second kind (see~\eqref{stirlinggen425}). 
Therefore, we have, for $k_1,...,k_m\geq 1$, $n\geq 1$ and $\ell\geq 0$, that
\begin{align}
&{\rm ad}_{L_{2k_m-1,2k_m}} \cdots {\rm ad}_{L_{2k_1-1,2k_1}}L_{n,2\ell}\notag\\
& = \frac{\prod_{j=1}^m (2k_j)!}{ \prod_{j=1}^m \binom{n+1-j+2\sum_{i=1}^j k_i}{n+2-j+2\sum_{i=1}^{j-1} k_i} } 
\sum_{h_1=k_1+\ell}^\infty \sum_{h_2=k_2+h_1}^\infty \dots \sum_{h_m=k_m+h_{m-1}}^\infty \nn\\
& \qquad \prod_{j=1}^m \left\{\begin{array}{c} 2h_j-2h_{j-1} \\ 2k_j-1\end{array}\right\} 
\prod_{j=1}^m \binom{2h_j}{2h_{j-1}} L_{n+(2k_1-1)+\cdots+(2k_m-1),2h},\notag
\end{align}
where it is understood that $h_0=\ell$.
The theorem is proved.
\end{proof}

For the case that $M$ is the one-dimensional Frobenius manifold with the Frobenius potential $F = (v^1)^3/6$.  We have
\beq
L_{-1}^{\rm H} =  L_{-1,0} + \frac{\sigma_1}{24}.
\eeq
We also have the following alternative explicit expression for the Virasoro operators~$L_n^{\rm H}$, $n \geq 0$:
\begin{equation}
\begin{split}
L_n^{\rm H} = & L_{n,0} + \sum_{m=1}^{n+1} \frac{(-1)^m}{m!} \sum_{k_1, \dots, k_m \geq 1}
\prod_{i=1}^m \left(\sigma_{k_i} \frac{B_{2k_i}}{(2k_i)!}\right) \\
& \times \biggl( \sum_{\ell=0}^\infty \left(\delta_{2k_m-1} \cdots \delta_{2k_1-1} A(\ell)\right)
 \tilde t_\ell \frac{\p}{\p t_{\ell+n+\sum_{i=1}^m (2k_i-1)}} \\
& \; \; - \frac{\e^2}{2} \sum_{a=0}^{n-1+ \sum_{i=1}^m (2k_i-1)}
(-1)^a \left(\delta_{2k_m-1} \cdots \delta_{2k_1-1} A(-a-1)\right) \\
& \; \; \times \frac{\p^2}{\p t_{a}\p t_{n-1+\sum_{i=1}^m (2k_i-1)-a}} \biggr),
\end{split}
\end{equation}
where $A(\ell) := \prod_{i=0}^n (\ell+i+\half)$ and $\delta_a$ denotes the difference operator $\delta_a A(x):= A(x+a)-A(x)$.
See~\cite{Zhou} for the proof.

\bibliographystyle{amsplain}

\medskip
\medskip

\noindent Si-Qi Liu

\noindent Department of Mathematical Sciences, Tsinghua University, 

\noindent Beijing 100084, P.R. China 

\noindent liusq@mail.tsinghua.edu.cn

\medskip
\medskip

\noindent Di Yang

\noindent School of Mathematical Sciences, University of Science and Technology of China,

\noindent Hefei 230026, P.R. China 

\noindent diyang@ustc.edu.cn

\medskip
\medskip

\noindent Youjin Zhang

\noindent Department of Mathematical Sciences, Tsinghua University, 

\noindent Beijing 100084, P.R. China 

\noindent youjin@mail.tsinghua.edu.cn

\medskip
\medskip

\noindent Jian Zhou

\noindent Department of Mathematical Sciences, Tsinghua University, 

\noindent Beijing 100084, P.R. China 

\noindent jianzhou@mail.tsinghua.edu.cn

\end{document}